\documentclass[12pt]{article}
\usepackage[latin1]{inputenc}
\usepackage{lmodern}

\usepackage{amssymb, amsmath, amsthm}
\usepackage[a4paper,top=25mm,bottom=25mm,left=25mm,right=25mm]{geometry}
\usepackage{etex}

\usepackage{authblk} 
\usepackage{pifont}
\usepackage{graphicx}
\usepackage[usenames,dvipsnames,svgnames,table]{xcolor}
\usepackage[figuresright]{rotating}
\usepackage{xtab} 
\usepackage{longtable} 
\usepackage{footnote}
\usepackage[stable]{footmisc}
\usepackage{chngpage} 
\usepackage{pdflscape} 

\usepackage{pgfplots}
\usepackage{setspace}

\makesavenoteenv{tabular}
\usepackage{tabularx}
\usepackage{booktabs}
\usepackage{multirow}
\usepackage{threeparttable}
\usepackage[referable]{threeparttablex} 
\newcolumntype{R}{>{\raggedleft\arraybackslash}X}
\newcolumntype{L}{>{\raggedright\arraybackslash}X}
\newcolumntype{C}{>{\centering\arraybackslash}X}
\newcolumntype{A}{>{\columncolor{gray!25}}C}
\newcolumntype{a}{>{\columncolor{gray!25}}c}

\usepackage{dcolumn} 
\newcolumntype{.}{D{.}{.}{-1}}

\usepackage{tikz}
\usetikzlibrary{arrows}
\usepackage[semicolon]{natbib}
\usepackage{hyperref} 
\usepackage{hyperref}
\hypersetup{
  colorlinks   = true,    
  urlcolor     = blue,    
  linkcolor    = blue,    
  citecolor    = red      
}

\usepackage{microtype}
\usepackage[justification=centerfirst]{caption}

\usepackage[labelformat=simple]{subcaption}

\DeclareCaptionLabelFormat{parenthesis}{(#2)}
\makeatletter
\renewcommand\p@subfigure{\arabic{figure}.}
\makeatother

\newenvironment{customlegend}[1][]{%
	\begingroup
	\csname pgfplots@init@cleared@structures\endcsname
	\pgfplotsset{#1}%
    }{%
	\csname pgfplots@createlegend\endcsname
	\endgroup
    }%
\def\addlegendimage{\csname pgfplots@addlegendimage\endcsname}

\usepackage{enumitem}

\setlist[itemize]{leftmargin=3\parindent}
\setlist[enumerate]{leftmargin=2\parindent}

\theoremstyle{plain}

\newtheorem{observation}{Observation}
\newtheorem{proposition}{Proposition}

\theoremstyle{definition}

\newtheorem{example}{Example}

\theoremstyle{remark}

\newtheorem{remark}{Remark}

\def\keywords{\vspace{.5em} 
{\textit{Keywords}:\,\relax%
}}

\def\JEL{\vspace{.5em} 
{\textbf{\emph{JEL} classification number}:\,\relax%
}}

\def\AMS{\vspace{.5em} 
{\textbf{AMS classification number}:\,}}

\author{L\'aszl\'o Csat\'o\thanks{~e-mail: laszlo.csato@uni-corvinus.hu} }
\affil{Department of Operations Research and Actuarial Sciences \\ Corvinus University of Budapest \\ MTA-BCE 'Lend\"ulet' Strategic Interactions Research Group \\ Budapest, Hungary}
\title{Between plurality and proportionality: an analysis of vote transfer systems\thanks{~We are grateful to Daniel Bochsler, Tam\'as Solymosi, Bal\'azs Sziklai and Attila Tasn\'adi for their comments and suggestions. \newline
The research was supported by OTKA grant K 111797.}}
\date{\today}

\begin{document}

\maketitle

\begin{abstract}
The paper considers a general model of electoral systems combining district-based elections with a compensatory mechanism in order to create any outcome between strictly majoritarian and purely proportional seat allocation. It contains vote transfer and allows for the application of three different correction formulas. Analysis in a two-party system shows that a trade-off exists for the dominant party between its expected seat share and its chance of obtaining majority. Vote transfer rules are also investigated by focusing on the possibility of manipulation.

The model is applied to the 2014 Hungarian parliamentary election. Hypothetical results reveal that the vote transfer rule cannot be evaluated in itself, only together with the share of constituency seats. With an appropriate choice of the latter, the three mechanisms may be functionally equivalent.

\end{abstract}

\JEL{D72}

\AMS{91B12}

\section{Introduction}

Mixed-member electoral systems combine direct election of representatives with the aim of a proportional (or, at least more proportional) seat allocation. In this category, most of the attention is devoted to mixed-member proportional (MMP) electoral systems \citep{ShugartWattenberg2001}, while vote transfer systems -- applied in Hungary (from 1990) and Italy (from 1993 to 2005) -- receive much less consideration.
Such systems are based on single-member districts and party lists, but they contain some compensation mechanism (through the transfer of votes that do not yield candidates or votes in excess of what is needed to win a seat) in order to avoid seemingly unfair outcomes.
They also have two advantages over common MMP rules: vote transfer systems seem to be more simple and intuitive (for example, there is no need for overhang seats) and they are immune to some kind of strategic manipulation such as tactical voting (vote-splitting) \citep{Bochsler2014, Bochsler2015}. Hungary has used vote transfer system over seven national elections but no such party strategies have emerged.

Vote transfer systems can be classified by their compensation mechanisms, which may significantly influence both the seat allocation and the electoral strategies of the parties. Hungarian electoral rules were fundamentally rewritten in 2012, after the governing party alliance FIDESZ-KDNP won a two-thirds (super)majority in the 2010 election: the number of single-member districts were reduced (this tier has been evaluated by \citet{BiroKoczySziklai2015}), and the complicated proportional representation pillar was simplified.
While vote transfer mechanism remained an essential part of the system, their calculation was modified, implying some interesting questions like:
Was the change favorable for the party that initiated the reform?
Is it possible to compare vote transfer systems without the use of former election results?
What are the advantages and disadvantages of different vote transfer formulas?

This paper attempts to address these issues by defining a model of vote transfer systems, which leads to a seat allocation between pure proportional and strictly majoritarian outcomes from a formal, theoretical perspective. The transition between the two extremities is governed by the ratio of compensatory mandates. A similar model is provided by \citet{Bochsler2014}, but it assess the proportionality of positive vote transfer systems and do not deal with the study of different vote transfer formulas.
Besides the quantitative analysis, an empirical comparison is possible on the basis of the results of the Hungarian parliamentary election in 2014. Similar hypothetical scenarios have been investigated by \citet{OrtonaOttonePonzano2008}, for example.

The paper is structured as follows. Section \ref{Sec2} presents the model and Section \ref{Sec3} discusses its properties in a two-party system. Since there is no chance to derive robust results in the case of a more complicated party structure, Section \ref{Sec4} compares the formal model with the Hungarian electoral system, and Section \ref{Sec5} plugs the results of the Hungarian parliamentary election in 2014 into the model in order to derive alternative outcomes. Finally, Section \ref{Sec6} summarizes the main findings in three observations.

\section{The model} \label{Sec2}

Vote transfer systems implement seat allocation in two tiers. One part of mandates are allocated in single-member constituencies under majority (first-past-the-post) rule, while compensatory seats serve for reaching proportionality. Votes cast for candidates who failed to become elected may not be wasted as transferred to the second tier. Theoretically, this system is able to deliver a proportional seat allocation, however, it depends not only on the weight of the second tier but on the behavior of parties and voters \citep{Bochsler2014}. It may occur that the same electoral rule results in an over-representation of large parties in some situations, but leads to an over-representation of small parties in other cases.

Thus vote transfer systems are usually not wholly proportional in practice. We think it is not necessarily against the intention of the makers of electoral rules: they rather want to create a parliament without too fragmented party structure (a potential failure of truly proportional systems), and, at the same time, avoid the domination of majority.

A mathematical model of general vote transfer systems should be able to give any seat allocation a seat allocation between pure proportional and strictly majoritarian electoral systems.
It can be achieved as follows:
\begin{itemize}[label=$\bullet$]
\item
We have a mixed-member electoral system with single-member constituencies and list seats. All voter has only one vote.\footnote{~Another interpretation can be that there are two votes on separate ballots for each voter. However, it leads to the emergence of vote-splitting.} All local candidates are associated with a party. Voters are identified by their party vote.
\item
In single-member constituencies the candidate of the party with most votes wins. Districts have an equal size.
\item
List seats are allocated according to the proportional rule. The model is continuous, there is no need for special apportionment rules.
\item
The allocation of list seats is based on the aggregated number of votes for each party's candidates plus optional correction votes from single-member constituencies. Three different transfer formulas are investigated:
\begin{enumerate}[label=\alph*)]
\item
\emph{Direct vote transfer} (DVT): there are no correction votes;
\item
\emph{Positive vote transfer} (PVT): losing candidates of parties in constituencies add these votes to their own list votes;
\item
\emph{Negative vote transfer} (NVT): votes that are not 'used' in constituencies are added to list votes, i.e. not only the votes for losing candidates but the votes of the winning candidate minus the votes of the second candidate.\footnote{~Direct refers to the feature that votes are converted into list votes without correction. Positive and negative vote transfer are more common expressions used by \citet{Bochsler2014}, among others.}
\end{enumerate}
\item
Constituencies and list seats can be combined by an arbitrary ratio. Share of constituency seats is $\alpha \in \left[ 0,1 \right]$.
\item
There exists no threshold for the parties to pass in order to be eligible for list seats.
\end{itemize}

Example \ref{Examp1} illustrates this electoral system.

\begin{example} \label{Examp1}
Let two parties ($A$ and $B$) and two districts be in the country with the following vote distribution:
\begin{center}
\begin{tabularx}{0.6\textwidth}{LCC} \toprule
          & Party $A$     & Party $B$ \\ \midrule
    Constituency 1 & 65\%  & 35\% \\
    Constituency 2 & 45\%  & 55\% \\ \midrule
    National & 55\%  & 45\% \\ \bottomrule
\end{tabularx}
\end{center}

Note that national support for the parties is the average of their vote share in the constituencies since districts have the same number of voters. Each party has won one direct mandate, party $A$ has obtained Constituency 1 and party $B$ has gained Constituency 2.

Let $\alpha = 0.6$. Then seat allocation according to the three transfer formulas is as follows:
\begin{center}
    \begin{tabularx}{\textwidth}{llCC} \toprule
    Method      & Mandate share in & Party $A$ & Party $B$ \\ \midrule
    \multirow{4}{*}{DVT}    & Constituencies & $50\%$      & $50\%$ \\
          & List -- direct & $55\%$      & $45\%$ \\
          & List -- normalized  & $55\%$      & $45\%$ \\
          & \textbf{Total} & $\mathbf{52}$\textbf{\%}      & $\mathbf{48}$\textbf{\%} \\ \midrule
    \multirow{5}{*}{PVT}    & Constituencies & $50\%$      & $50\%$ \\
          & List -- direct & $55\%$      & $45\%$ \\
          & List -- losing votes & $0.5 \times 45\% = 22.5\%$      & $0.5 \times 35\% = 17.5\%$ \\
          & List -- normalized  & $77.5\% / 140\% \approx 55.36\%$      & $62.5\% / 140\% \approx 44.64\%$ \\
          & \textbf{Total} & $\mathbf{\approx 52.14}$\textbf{\%}      & $\mathbf{\approx 47.86}$\textbf{\%} \\ \midrule
    \multirow{6}{*}{NVT}    & Constituencies & $50\%$      & $50\%$ \\
          & List -- direct & $55\%$      & $45\%$ \\
          & List -- losing votes & $0.5 \times 45\% = 22.5\%$      & $0.5 \times 35\% = 17.5\%$ \\
          & List -- winning votes & $0.5 \times (65\%-35\%) = 15\%$      & $0.5 \times (55\%-45\%) = 5\%$ \\
          & List -- normalized  & $92.5\% / 160\% = 57.8125\%$      & $67.5\% / 160\% = 42.1875\%$ \\
          & \textbf{Total} & $\mathbf{53.125}$\textbf{\%}      & $\mathbf{46.875}$\textbf{\%} \\ \bottomrule
\end{tabularx}
\end{center}

Calculation of list votes may require some comments. In the case of DVT, it is simply the number of votes across all constituencies. In the case of PVT, some correction (indirect) votes are added to this pool, namely, votes in districts that do not win a seat, which is $0.5 \times 35\%$ for party $B$ in constituency 1 and $0.5 \times 45\%$ for party $A$ in constituency 2 as both districts contain the half of all voters. Thus the total number of list votes exceeds the number of voters, so they should be normalized.
In the case of NVT, calculation of correction votes is more complicated since all 'unused' votes count. For example, it is $0.5 \times 35\%$ for party $B$ (the loser) and $0.5 \times (65\% - 35\%)$ for party $A$ (the winner) in Constituency 1.

Constituency and normalized list vote shares are aggregated as a weighted sum governed by parameter $\alpha$. For instance, party $B$ obtains $0.6 \times 50\% + 0.4 \times 42.1875\% = 46.875\%$ of mandates under NVT. 
\end{example}

Note that party $A$ benefits from a smaller $\alpha$ in all cases since it is under-represented in the share of constituencies won.
List votes somewhat adjusts this disproportion and approximate to the national share of votes. DVT, PVT and NVT gradually increase the role of aggregated votes, i.e. the ratio of parliamentary seats is closer to 55:45\%. However, their complexity also rises: in DVT, list votes are not influenced by the distribution of votes among the constituencies; in PVT, list votes depend on the number of losing votes, that is, which party has won; and in NVT, list votes are affected by both the winner and the second party.

\section{Study of the model} \label{Sec3}

This section investigates the model in the case of two parties $A$ and $B$, and discusses the possibility to generalize the results as well as the potential of different electoral strategies.

\subsection{Analytical results} \label{Sec31}

Let the vote share of party $A$ be fixed at $x \in \left( 0.5; \, 1 \right]$. Therefore it will be called to the \emph{dominant} party and the other party $B$ is called to the \emph{inferior} party. The result is deterministic if $\alpha = 0$. However, voters of party $A$ may be distributed arbitrarily among the constituencies.
In two extreme cases the outcome can be derived mathematically.

The first is when party $A$ has the opportunity to gerrymander district boundaries. Since it has more than 50\% of votes at the national level, all constituencies will be won by party $A$. List votes are as follows:
\begin{center}
	\begin{tabularx}{\textwidth}{lCC} \toprule
          & Party $A$ & Party $B$ \\ \midrule
    Direct & $x$     & $1-x$ \\
    Losing votes & $0$     & $1-x$ \\
    Winning votes & $x-(1-x)=2x-1$ & $0$ \\ \midrule
    List normalized (DVT) & $x$      & $1-x$ \\
    List normalized (PVT) & $x / (2-x)$    	& $(2-2x) / (2-x)$ \\
    List normalized (NVT) & $(3x-1) / (1+x)$   & $(2-2x) / (1+x)$ \\ \bottomrule
	\end{tabularx}
\end{center}

It implies the following result.

\begin{proposition} \label{Prop1}
Consider a two-party system where the dominant party has a gerrymandering power.
The dominant party's preference order is DVT $\succ$ NVT $\succ$ PVT.
DVT and NVT are always more favorable for the dominant party than a proportional rule, while PVT is better if the ratio of constituency seats is at least the half of its vote share.
\end{proposition}

\begin{proof}
Note that
\[
x \geq \frac{x}{1 + (1-x)} = \frac{x}{2-x} \quad \text{and} \quad \frac{3x-1}{1+x} = \frac{x + (2x-1)}{(2-x) + (2x-1)} > \frac{x}{2-x},
\]
so DVT and NVT are more favorable for party $A$ than PVT. Moreover,
\[
x \geq \frac{3x-1}{1+x} \Leftrightarrow (x-1)^2 \geq 0,
\]
hence party $A$ has the preference order DVT $\succeq$ NVT $\succ$ PVT with a strict relation if $x < 1$.

According to DVT, $\alpha + (1-\alpha)x \geq x$. By applying PVT:
\[
\alpha + (1-\alpha) \frac{x}{2-x} \geq x \Leftrightarrow x^2 - (1+2\alpha)x + 2\alpha \geq 0 \Leftrightarrow \left[ x \geq 1 \quad \text{or} \quad x \leq 2\alpha \right].
\]
It means that the share of constituency seats should be at least the half of party $A$'s vote share to provide that it benefits from the system compared to the proportional rule.
By applying NVT:
\[
\alpha + (1-\alpha) \frac{3x+1}{1+x} \geq x \Leftrightarrow 0 \geq x^2 - (2-2\alpha)x + (1-2\alpha) \Leftrightarrow 1 \geq x \geq -2\alpha.
\]
NVT is always more favorable for party $A$ than a proportional rule.
\end{proof}

The second extremity is when party $B$ has the power to draw district boundaries in order to win them marginally. Then the ratio of constituencies won is $2(1-x)$ as party $A$ has fewer voters in each of them than party $B$. List votes are as follows:
\begin{center}
	\begin{tabularx}{\textwidth}{lCC} \toprule
          & Party $A$ & Party $B$ \\ \midrule
    Direct & $x$     & $1-x$ \\
    Losing votes & $1-x$     & $0$ \\
    Winning votes & $x-(1-x)=2x-1$ & $0$ \\ \midrule
    List normalized (DVT) & $x$      & $1-x$ \\
    List normalized (PVT) & $1 / (2-x)$    	& $(1-x) / (2-x)$ \\
    List normalized (NVT) & $2x / (1+x)$   & $(1-x) / (1+x)$ \\ \bottomrule
	\end{tabularx}
\end{center}

It implies the following result, analogous to Proposition \ref{Prop1}.

\begin{proposition} \label{Prop2}
Consider a two-party system where the inferior party has a gerrymandering power.
The dominant party's preference order is NVT $\succ$ PVT $\succ$ DVT.
DVT is never more favorable for the dominant party than a proportional rule, while PVT and NVT are better only if the share of constituency seats is small (one-third is an upper limit).
\end{proposition}

\begin{proof}
Note that
\[
x \leq \frac{x + (1-x)}{1 + (1-x)} = \frac{1}{2-x} \leq \frac{1 + (2x-1)}{2-x + (2x-1)} = \frac{2x}{1 + x},
\]
so party $A$ has the preference order NVT $\succeq$ PVT $\succeq$ DVT with a strict relation if $x < 1$.

According to DVT, $\alpha(2x-1) + (1-\alpha)x \leq x$ as $2x-1 \leq x$. By applying PVT:
\[
\alpha (2x-1) + (1-\alpha) \frac{1}{2-x} \geq x \Leftrightarrow (1-2\alpha)x^2 - (2-5\alpha)x + (1-3\alpha) \geq 0 \Leftrightarrow
\]
\[
\Leftrightarrow \left[ x \geq 1 \quad \text{or} \quad x \leq \frac{1-3\alpha}{1-2\alpha} \right].
\]
It means that 
\[
\alpha \leq \frac{1-x}{3-2x} = \frac{1}{3} - \frac{x}{9-6x} \leq \frac{1}{3} - \frac{1/2}{9-3} = \frac{1}{4},
\]
thus PVT is certainly not more favorable for party $A$ than a proportional rule if the share of constituency seats is at least 25\%.
By applying NVT:
\[
\alpha (2x-1) + (1-\alpha) \frac{2x}{1+x} \geq x \Leftrightarrow 0 \geq (1-2\alpha)x^2 - (1-\alpha)x + \alpha \Leftrightarrow 1 \geq x \geq \frac{\alpha}{1-2\alpha}.
\]
Hence NVT is more favorable for party $A$ than a proportional rule if $\alpha \leq 1/2$ and
\[
\alpha \leq \frac{x}{1+2x} = \frac{1}{2} - \frac{1}{2+4x} \leq \frac{1}{2} - \frac{1}{2+4} = \frac{1}{3}.
\]
\end{proof}

\begin{remark} \label{Rem1}
Consider a two-party system where the inferior party has a gerrymandering power.
The necessary condition for $\alpha$ to provide that the dominant party benefits from the system compared to the proportional rule is decreasing in $x$ for PVT (but never exceeds $1/4$) and increasing for NVT (but never exceeds $1/3$).
\end{remark}

\begin{figure}[htbp]
\centering
\caption{Seat-vote ratios, extreme cases}
\label{Fig1}

\begin{subfigure}{\textwidth}
	\centering
	\caption{$\alpha = 0.3$}
	
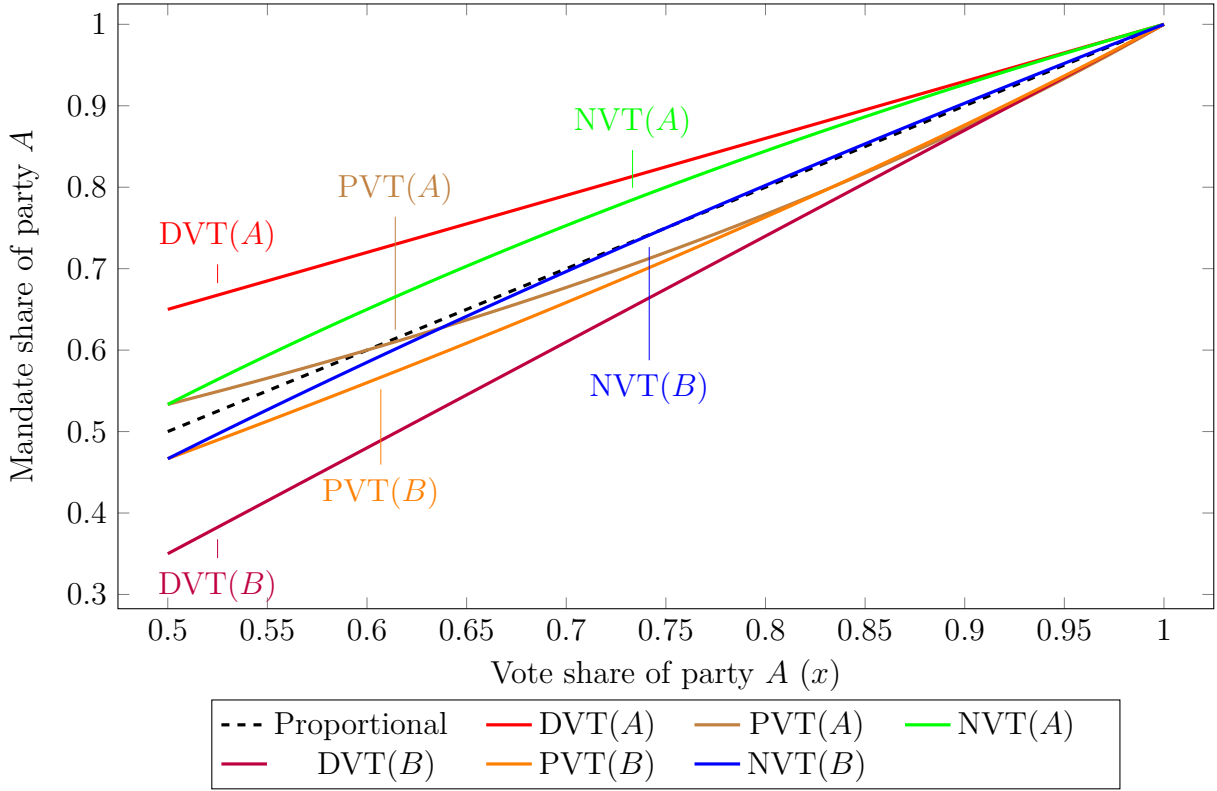
\label{Fig1a}
	
\begin{tikzpicture}
\begin{axis}[width=\textwidth, 
height=0.6\textwidth,
xmin=0.475,
xmax=1.025,
ymax=1.025,
xlabel = Vote share of party $A$ ($x$),
ylabel = Mandate share of party $A$,
legend entries={Proportional$\quad$,DVT($A$)$\quad$,PVT($A$)$\quad$,NVT($A$)$\quad$,DVT($B$),PVT($B$)$\quad$,NVT($B$)$\quad$},
legend style={at={(0.5,-0.15)},anchor = north,legend columns = 4}
]

\addplot[black,dashed,smooth,very thick] coordinates {
(0.5,0.5)
(0.51,0.51)
(0.52,0.52)
(0.53,0.53)
(0.54,0.54)
(0.55,0.55)
(0.56,0.56)
(0.57,0.57)
(0.58,0.58)
(0.59,0.59)
(0.6,0.6)
(0.61,0.61)
(0.62,0.62)
(0.63,0.63)
(0.64,0.64)
(0.65,0.65)
(0.66,0.66)
(0.67,0.67)
(0.68,0.68)
(0.69,0.69)
(0.7,0.7)
(0.71,0.71)
(0.72,0.72)
(0.73,0.73)
(0.74,0.74)
(0.75,0.75)
(0.76,0.76)
(0.77,0.77)
(0.78,0.78)
(0.79,0.79)
(0.8,0.8)
(0.81,0.81)
(0.82,0.82)
(0.83,0.83)
(0.84,0.84)
(0.85,0.85)
(0.86,0.86)
(0.87,0.87)
(0.88,0.88)
(0.89,0.89)
(0.9,0.9)
(0.91,0.91)
(0.92,0.92)
(0.93,0.93)
(0.94,0.94)
(0.95,0.95)
(0.96,0.96)
(0.97,0.97)
(0.98,0.98)
(0.99,0.99)
(1,1)
};

\addplot[red,smooth,very thick] coordinates {
(0.5,0.65)
(0.51,0.657)
(0.52,0.664)
(0.53,0.671)
(0.54,0.678)
(0.55,0.685)
(0.56,0.692)
(0.57,0.699)
(0.58,0.706)
(0.59,0.713)
(0.6,0.72)
(0.61,0.727)
(0.62,0.734)
(0.63,0.741)
(0.64,0.748)
(0.65,0.755)
(0.66,0.762)
(0.67,0.769)
(0.68,0.776)
(0.69,0.783)
(0.7,0.79)
(0.71,0.797)
(0.72,0.804)
(0.73,0.811)
(0.74,0.818)
(0.75,0.825)
(0.76,0.832)
(0.77,0.839)
(0.78,0.846)
(0.79,0.853)
(0.8,0.86)
(0.81,0.867)
(0.82,0.874)
(0.83,0.881)
(0.84,0.888)
(0.85,0.895)
(0.86,0.902)
(0.87,0.909)
(0.88,0.916)
(0.89,0.923)
(0.9,0.93)
(0.91,0.937)
(0.92,0.944)
(0.93,0.951)
(0.94,0.958)
(0.95,0.965)
(0.96,0.972)
(0.97,0.979)
(0.98,0.986)
(0.99,0.993)
(1,1)
}
node [pos=0.05,pin={[pin edge={red}, pin distance=0.25cm] 90:{DVT($A$)}}] {};

\addplot[brown,smooth,very thick] coordinates {
(0.5,0.533333333333333)
(0.51,0.539597315436242)
(0.52,0.545945945945946)
(0.53,0.552380952380952)
(0.54,0.558904109589041)
(0.55,0.56551724137931)
(0.56,0.572222222222222)
(0.57,0.579020979020979)
(0.58,0.585915492957746)
(0.59,0.59290780141844)
(0.6,0.6)
(0.61,0.607194244604316)
(0.62,0.614492753623188)
(0.63,0.621897810218978)
(0.64,0.629411764705882)
(0.65,0.637037037037037)
(0.66,0.644776119402985)
(0.67,0.652631578947368)
(0.68,0.660606060606061)
(0.69,0.668702290076336)
(0.7,0.676923076923077)
(0.71,0.685271317829457)
(0.72,0.69375)
(0.73,0.702362204724409)
(0.74,0.711111111111111)
(0.75,0.72)
(0.76,0.729032258064516)
(0.77,0.738211382113821)
(0.78,0.747540983606557)
(0.79,0.75702479338843)
(0.8,0.766666666666667)
(0.81,0.776470588235294)
(0.82,0.786440677966102)
(0.83,0.796581196581197)
(0.84,0.806896551724138)
(0.85,0.817391304347826)
(0.86,0.828070175438596)
(0.87,0.838938053097345)
(0.88,0.85)
(0.89,0.861261261261261)
(0.9,0.872727272727273)
(0.91,0.884403669724771)
(0.92,0.896296296296296)
(0.93,0.908411214953271)
(0.94,0.920754716981132)
(0.95,0.933333333333333)
(0.96,0.946153846153846)
(0.97,0.959223300970874)
(0.98,0.972549019607843)
(0.99,0.986138613861386)
(1,1)
}
node [pos=0.2,pin={[pin edge={brown}, pin distance=1.5cm] 90:{PVT($A$)}}] {};

\addplot[green,smooth,very thick] coordinates {
(0.5,0.533333333333333)
(0.51,0.545695364238411)
(0.52,0.557894736842105)
(0.53,0.569934640522876)
(0.54,0.581818181818182)
(0.55,0.593548387096774)
(0.56,0.605128205128205)
(0.57,0.61656050955414)
(0.58,0.627848101265823)
(0.59,0.638993710691824)
(0.6,0.65)
(0.61,0.660869565217391)
(0.62,0.671604938271605)
(0.63,0.682208588957055)
(0.64,0.692682926829268)
(0.65,0.703030303030303)
(0.66,0.713253012048193)
(0.67,0.723353293413174)
(0.68,0.733333333333333)
(0.69,0.743195266272189)
(0.7,0.752941176470588)
(0.71,0.762573099415205)
(0.72,0.772093023255814)
(0.73,0.78150289017341)
(0.74,0.790804597701149)
(0.75,0.8)
(0.76,0.809090909090909)
(0.77,0.818079096045198)
(0.78,0.826966292134831)
(0.79,0.835754189944134)
(0.8,0.844444444444445)
(0.81,0.853038674033149)
(0.82,0.861538461538462)
(0.83,0.869945355191257)
(0.84,0.878260869565217)
(0.85,0.886486486486486)
(0.86,0.894623655913978)
(0.87,0.902673796791444)
(0.88,0.91063829787234)
(0.89,0.918518518518518)
(0.9,0.926315789473684)
(0.91,0.934031413612565)
(0.92,0.941666666666667)
(0.93,0.949222797927461)
(0.94,0.956701030927835)
(0.95,0.964102564102564)
(0.96,0.971428571428571)
(0.97,0.978680203045685)
(0.98,0.985858585858586)
(0.99,0.992964824120603)
(1,1)
}
node [pos=0.5,pin={[pin edge={green}, pin distance=0.5cm] 90:{NVT($A$)}}] {};

\addplot[purple,smooth,very thick] coordinates {
(0.5,0.35)
(0.51,0.363)
(0.52,0.376)
(0.53,0.389)
(0.54,0.402)
(0.55,0.415)
(0.56,0.428)
(0.57,0.441)
(0.58,0.454)
(0.59,0.467)
(0.6,0.48)
(0.61,0.493)
(0.62,0.506)
(0.63,0.519)
(0.64,0.532)
(0.65,0.545)
(0.66,0.558)
(0.67,0.571)
(0.68,0.584)
(0.69,0.597)
(0.7,0.61)
(0.71,0.623)
(0.72,0.636)
(0.73,0.649)
(0.74,0.662)
(0.75,0.675)
(0.76,0.688)
(0.77,0.701)
(0.78,0.714)
(0.79,0.727)
(0.8,0.74)
(0.81,0.753)
(0.82,0.766)
(0.83,0.779)
(0.84,0.792)
(0.85,0.805)
(0.86,0.818)
(0.87,0.831)
(0.88,0.844)
(0.89,0.857)
(0.9,0.87)
(0.91,0.883)
(0.92,0.896)
(0.93,0.909)
(0.94,0.922)
(0.95,0.935)
(0.96,0.948)
(0.97,0.961)
(0.98,0.974)
(0.99,0.987)
(1,1)
}
node [pos=0.05,pin={[pin edge={purple}, pin distance=0.25cm] 270:{DVT($B$)}}] {};

\addplot[orange,smooth,very thick] coordinates {
(0.5,0.466666666666667)
(0.51,0.475798657718121)
(0.52,0.484972972972973)
(0.53,0.494190476190476)
(0.54,0.50345205479452)
(0.55,0.512758620689655)
(0.56,0.522111111111111)
(0.57,0.531510489510489)
(0.58,0.540957746478873)
(0.59,0.55045390070922)
(0.6,0.56)
(0.61,0.569597122302158)
(0.62,0.579246376811594)
(0.63,0.588948905109489)
(0.64,0.598705882352941)
(0.65,0.608518518518518)
(0.66,0.618388059701493)
(0.67,0.628315789473684)
(0.68,0.63830303030303)
(0.69,0.648351145038168)
(0.7,0.658461538461538)
(0.71,0.668635658914729)
(0.72,0.678875)
(0.73,0.689181102362205)
(0.74,0.699555555555555)
(0.75,0.71)
(0.76,0.720516129032258)
(0.77,0.731105691056911)
(0.78,0.741770491803279)
(0.79,0.752512396694215)
(0.8,0.763333333333333)
(0.81,0.774235294117647)
(0.82,0.785220338983051)
(0.83,0.796290598290598)
(0.84,0.807448275862069)
(0.85,0.818695652173913)
(0.86,0.830035087719298)
(0.87,0.841469026548673)
(0.88,0.853)
(0.89,0.864630630630631)
(0.9,0.876363636363636)
(0.91,0.888201834862385)
(0.92,0.900148148148148)
(0.93,0.912205607476636)
(0.94,0.924377358490566)
(0.95,0.936666666666667)
(0.96,0.949076923076923)
(0.97,0.961611650485437)
(0.98,0.974274509803921)
(0.99,0.987069306930693)
(1,1)
}
node [pos=0.2,pin={[pin edge={orange}, pin distance=1cm] 270:{PVT($B$)}}] {};

\addplot[blue,smooth,very thick] coordinates {
(0.5,0.466666666666667)
(0.51,0.478847682119205)
(0.52,0.490947368421053)
(0.53,0.502967320261438)
(0.54,0.514909090909091)
(0.55,0.526774193548387)
(0.56,0.538564102564103)
(0.57,0.55028025477707)
(0.58,0.561924050632911)
(0.59,0.573496855345912)
(0.6,0.585)
(0.61,0.596434782608696)
(0.62,0.607802469135802)
(0.63,0.619104294478528)
(0.64,0.630341463414634)
(0.65,0.641515151515151)
(0.66,0.652626506024096)
(0.67,0.663676646706587)
(0.68,0.674666666666667)
(0.69,0.685597633136095)
(0.7,0.696470588235294)
(0.71,0.707286549707602)
(0.72,0.718046511627907)
(0.73,0.728751445086705)
(0.74,0.739402298850575)
(0.75,0.75)
(0.76,0.760545454545455)
(0.77,0.771039548022599)
(0.78,0.781483146067416)
(0.79,0.791877094972067)
(0.8,0.802222222222222)
(0.81,0.812519337016575)
(0.82,0.822769230769231)
(0.83,0.832972677595628)
(0.84,0.843130434782609)
(0.85,0.853243243243243)
(0.86,0.863311827956989)
(0.87,0.873336898395722)
(0.88,0.88331914893617)
(0.89,0.893259259259259)
(0.9,0.903157894736842)
(0.91,0.913015706806283)
(0.92,0.922833333333333)
(0.93,0.932611398963731)
(0.94,0.942350515463917)
(0.95,0.952051282051282)
(0.96,0.961714285714286)
(0.97,0.971340101522842)
(0.98,0.980929292929293)
(0.99,0.990482412060302)
(1,1)
}
node [pos=0.5,pin={[pin edge={blue}, pin distance=1.5cm] 270:{NVT($B$)}}] {};
\end{axis}
\end{tikzpicture}
\end{subfigure}
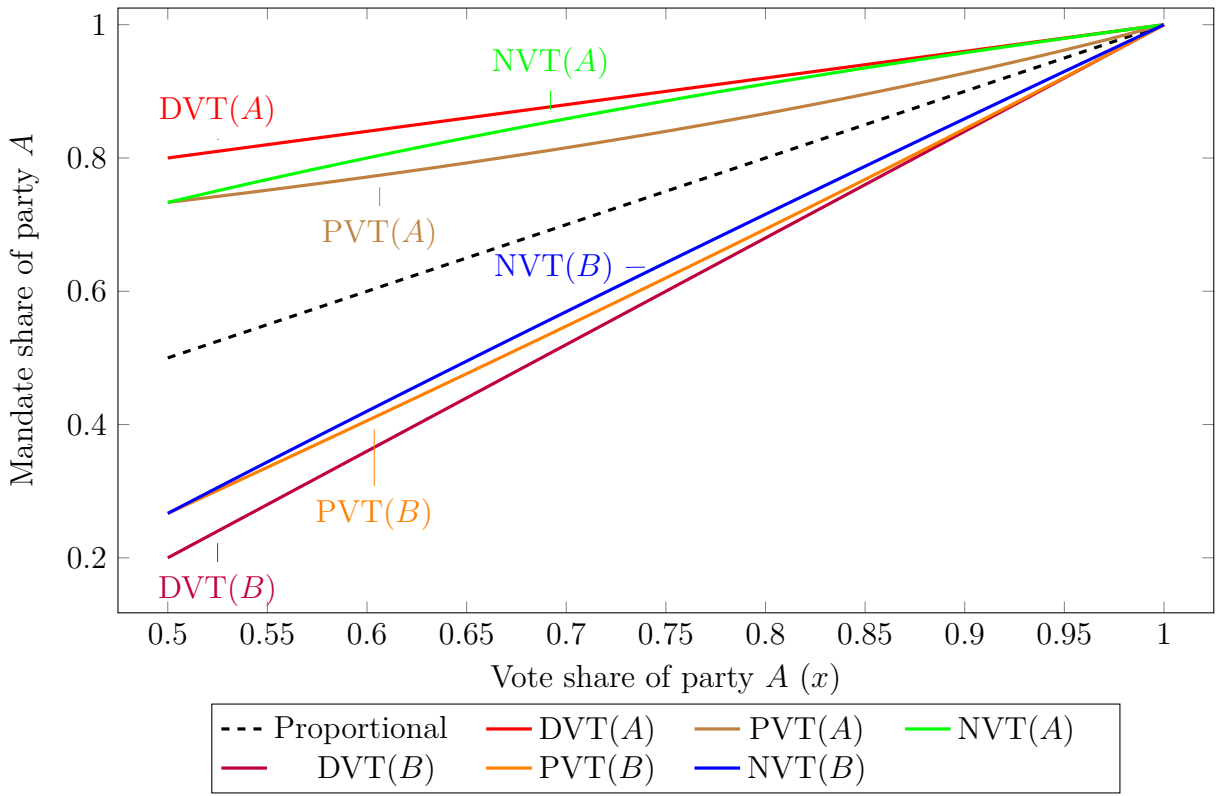
\vspace{0.5cm}

\begin{subfigure}{\textwidth}
	\centering
	\caption{$\alpha = 0.6$}
	\label{Fig1b}
	
\begin{tikzpicture}
\begin{axis}[width=\textwidth, 
height=0.6\textwidth,
xmin=0.475,
xmax=1.025,
ymax=1.025,
xlabel = Vote share of party $A$ ($x$),
ylabel = Mandate share of party $A$,
legend entries={Proportional$\quad$,DVT($A$)$\quad$,PVT($A$)$\quad$,NVT($A$)$\quad$,DVT($B$),PVT($B$)$\quad$,NVT($B$)$\quad$},
legend style={at={(0.5,-0.15)},anchor = north,legend columns = 4}
]

\addplot[black,dashed,smooth,very thick] coordinates {
(0.5,0.5)
(0.51,0.51)
(0.52,0.52)
(0.53,0.53)
(0.54,0.54)
(0.55,0.55)
(0.56,0.56)
(0.57,0.57)
(0.58,0.58)
(0.59,0.59)
(0.6,0.6)
(0.61,0.61)
(0.62,0.62)
(0.63,0.63)
(0.64,0.64)
(0.65,0.65)
(0.66,0.66)
(0.67,0.67)
(0.68,0.68)
(0.69,0.69)
(0.7,0.7)
(0.71,0.71)
(0.72,0.72)
(0.73,0.73)
(0.74,0.74)
(0.75,0.75)
(0.76,0.76)
(0.77,0.77)
(0.78,0.78)
(0.79,0.79)
(0.8,0.8)
(0.81,0.81)
(0.82,0.82)
(0.83,0.83)
(0.84,0.84)
(0.85,0.85)
(0.86,0.86)
(0.87,0.87)
(0.88,0.88)
(0.89,0.89)
(0.9,0.9)
(0.91,0.91)
(0.92,0.92)
(0.93,0.93)
(0.94,0.94)
(0.95,0.95)
(0.96,0.96)
(0.97,0.97)
(0.98,0.98)
(0.99,0.99)
(1,1)
};

\addplot[red,smooth,very thick] coordinates {
(0.5,0.8)
(0.51,0.804)
(0.52,0.808)
(0.53,0.812)
(0.54,0.816)
(0.55,0.82)
(0.56,0.824)
(0.57,0.828)
(0.58,0.832)
(0.59,0.836)
(0.6,0.84)
(0.61,0.844)
(0.62,0.848)
(0.63,0.852)
(0.64,0.856)
(0.65,0.86)
(0.66,0.864)
(0.67,0.868)
(0.68,0.872)
(0.69,0.876)
(0.7,0.88)
(0.71,0.884)
(0.72,0.888)
(0.73,0.892)
(0.74,0.896)
(0.75,0.9)
(0.76,0.904)
(0.77,0.908)
(0.78,0.912)
(0.79,0.916)
(0.8,0.92)
(0.81,0.924)
(0.82,0.928)
(0.83,0.932)
(0.84,0.936)
(0.85,0.94)
(0.86,0.944)
(0.87,0.948)
(0.88,0.952)
(0.89,0.956)
(0.9,0.96)
(0.91,0.964)
(0.92,0.968)
(0.93,0.972)
(0.94,0.976)
(0.95,0.98)
(0.96,0.984)
(0.97,0.988)
(0.98,0.992)
(0.99,0.996)
(1,1)
}
node [pos=0.05,pin={[pin edge={red}, pin distance=0cm] 90:{DVT($A$)}}] {};

\addplot[brown,smooth,very thick] coordinates {
(0.5,0.733333333333333)
(0.51,0.736912751677852)
(0.52,0.740540540540541)
(0.53,0.74421768707483)
(0.54,0.747945205479452)
(0.55,0.751724137931034)
(0.56,0.755555555555556)
(0.57,0.759440559440559)
(0.58,0.763380281690141)
(0.59,0.767375886524823)
(0.6,0.771428571428571)
(0.61,0.775539568345324)
(0.62,0.779710144927536)
(0.63,0.783941605839416)
(0.64,0.788235294117647)
(0.65,0.792592592592593)
(0.66,0.797014925373134)
(0.67,0.801503759398496)
(0.68,0.806060606060606)
(0.69,0.810687022900763)
(0.7,0.815384615384615)
(0.71,0.82015503875969)
(0.72,0.825)
(0.73,0.82992125984252)
(0.74,0.834920634920635)
(0.75,0.84)
(0.76,0.845161290322581)
(0.77,0.850406504065041)
(0.78,0.855737704918033)
(0.79,0.861157024793388)
(0.8,0.866666666666667)
(0.81,0.872268907563025)
(0.82,0.877966101694915)
(0.83,0.883760683760684)
(0.84,0.889655172413793)
(0.85,0.895652173913044)
(0.86,0.901754385964912)
(0.87,0.907964601769912)
(0.88,0.914285714285714)
(0.89,0.920720720720721)
(0.9,0.927272727272727)
(0.91,0.93394495412844)
(0.92,0.940740740740741)
(0.93,0.947663551401869)
(0.94,0.954716981132075)
(0.95,0.961904761904762)
(0.96,0.969230769230769)
(0.97,0.976699029126214)
(0.98,0.984313725490196)
(0.99,0.992079207920792)
(1,1)
}
node [pos=0.2,pin={[pin edge={brown}, pin distance=0.25cm] 270:{PVT($A$)}}] {};

\addplot[green,smooth,very thick] coordinates {
(0.5,0.733333333333333)
(0.51,0.740397350993377)
(0.52,0.747368421052632)
(0.53,0.754248366013072)
(0.54,0.761038961038961)
(0.55,0.767741935483871)
(0.56,0.774358974358974)
(0.57,0.780891719745223)
(0.58,0.787341772151899)
(0.59,0.793710691823899)
(0.6,0.8)
(0.61,0.806211180124224)
(0.62,0.812345679012346)
(0.63,0.81840490797546)
(0.64,0.824390243902439)
(0.65,0.83030303030303)
(0.66,0.836144578313253)
(0.67,0.841916167664671)
(0.68,0.847619047619048)
(0.69,0.853254437869822)
(0.7,0.858823529411765)
(0.71,0.864327485380117)
(0.72,0.869767441860465)
(0.73,0.87514450867052)
(0.74,0.880459770114942)
(0.75,0.885714285714286)
(0.76,0.890909090909091)
(0.77,0.896045197740113)
(0.78,0.901123595505618)
(0.79,0.906145251396648)
(0.8,0.911111111111111)
(0.81,0.916022099447514)
(0.82,0.920879120879121)
(0.83,0.92568306010929)
(0.84,0.930434782608696)
(0.85,0.935135135135135)
(0.86,0.939784946236559)
(0.87,0.944385026737968)
(0.88,0.948936170212766)
(0.89,0.953439153439153)
(0.9,0.957894736842105)
(0.91,0.962303664921466)
(0.92,0.966666666666667)
(0.93,0.970984455958549)
(0.94,0.975257731958763)
(0.95,0.979487179487179)
(0.96,0.983673469387755)
(0.97,0.987817258883249)
(0.98,0.991919191919192)
(0.99,0.995979899497487)
(1,1)
}
node [pos=0.4,pin={[pin edge={green}, pin distance=0.25cm] 90:{NVT($A$)}}] {};

\addplot[purple,smooth,very thick] coordinates {
(0.5,0.2)
(0.51,0.216)
(0.52,0.232)
(0.53,0.248)
(0.54,0.264)
(0.55,0.28)
(0.56,0.296)
(0.57,0.312)
(0.58,0.328)
(0.59,0.344)
(0.6,0.36)
(0.61,0.376)
(0.62,0.392)
(0.63,0.408)
(0.64,0.424)
(0.65,0.44)
(0.66,0.456)
(0.67,0.472)
(0.68,0.488)
(0.69,0.504)
(0.7,0.52)
(0.71,0.536)
(0.72,0.552)
(0.73,0.568)
(0.74,0.584)
(0.75,0.6)
(0.76,0.616)
(0.77,0.632)
(0.78,0.648)
(0.79,0.664)
(0.8,0.68)
(0.81,0.696)
(0.82,0.712)
(0.83,0.728)
(0.84,0.744)
(0.85,0.76)
(0.86,0.776)
(0.87,0.792)
(0.88,0.808)
(0.89,0.824)
(0.9,0.84)
(0.91,0.856)
(0.92,0.872)
(0.93,0.888)
(0.94,0.904)
(0.95,0.92)
(0.96,0.936)
(0.97,0.952)
(0.98,0.968)
(0.99,0.984)
(1,1)
}
node [pos=0.05,pin={[pin edge={purple}, pin distance=0.25cm] 270:{DVT($B$)}}] {};

\addplot[orange,smooth,very thick] coordinates {
(0.5,0.266666666666667)
(0.51,0.280456375838926)
(0.52,0.29427027027027)
(0.53,0.308108843537415)
(0.54,0.321972602739726)
(0.55,0.335862068965517)
(0.56,0.349777777777778)
(0.57,0.36372027972028)
(0.58,0.37769014084507)
(0.59,0.391687943262411)
(0.6,0.405714285714286)
(0.61,0.419769784172662)
(0.62,0.433855072463768)
(0.63,0.447970802919708)
(0.64,0.462117647058824)
(0.65,0.476296296296296)
(0.66,0.490507462686567)
(0.67,0.504751879699248)
(0.68,0.519030303030303)
(0.69,0.533343511450382)
(0.7,0.547692307692308)
(0.71,0.562077519379845)
(0.72,0.5765)
(0.73,0.59096062992126)
(0.74,0.605460317460317)
(0.75,0.62)
(0.76,0.63458064516129)
(0.77,0.64920325203252)
(0.78,0.663868852459016)
(0.79,0.678578512396694)
(0.8,0.693333333333333)
(0.81,0.708134453781513)
(0.82,0.722983050847458)
(0.83,0.737880341880342)
(0.84,0.752827586206896)
(0.85,0.767826086956522)
(0.86,0.782877192982456)
(0.87,0.797982300884956)
(0.88,0.813142857142857)
(0.89,0.82836036036036)
(0.9,0.843636363636364)
(0.91,0.85897247706422)
(0.92,0.87437037037037)
(0.93,0.889831775700935)
(0.94,0.905358490566038)
(0.95,0.920952380952381)
(0.96,0.936615384615385)
(0.97,0.952349514563107)
(0.98,0.968156862745098)
(0.99,0.984039603960396)
(1,1)
}
node [pos=0.2,pin={[pin edge={orange}, pin distance=0.75cm] 270:{PVT($B$)}}] {};

\addplot[blue,smooth,very thick] coordinates {
(0.5,0.266666666666667)
(0.51,0.282198675496689)
(0.52,0.297684210526316)
(0.53,0.313124183006536)
(0.54,0.328519480519481)
(0.55,0.343870967741936)
(0.56,0.359179487179487)
(0.57,0.374445859872611)
(0.58,0.389670886075949)
(0.59,0.40485534591195)
(0.6,0.42)
(0.61,0.435105590062112)
(0.62,0.450172839506173)
(0.63,0.46520245398773)
(0.64,0.480195121951219)
(0.65,0.495151515151515)
(0.66,0.510072289156627)
(0.67,0.524958083832335)
(0.68,0.539809523809524)
(0.69,0.554627218934911)
(0.7,0.569411764705882)
(0.71,0.584163742690058)
(0.72,0.598883720930232)
(0.73,0.61357225433526)
(0.74,0.628229885057471)
(0.75,0.642857142857143)
(0.76,0.657454545454546)
(0.77,0.672022598870057)
(0.78,0.686561797752809)
(0.79,0.701072625698324)
(0.8,0.715555555555556)
(0.81,0.730011049723757)
(0.82,0.74443956043956)
(0.83,0.758841530054645)
(0.84,0.773217391304348)
(0.85,0.787567567567568)
(0.86,0.80189247311828)
(0.87,0.816192513368984)
(0.88,0.830468085106383)
(0.89,0.844719576719577)
(0.9,0.858947368421053)
(0.91,0.873151832460733)
(0.92,0.887333333333333)
(0.93,0.901492227979275)
(0.94,0.915628865979381)
(0.95,0.92974358974359)
(0.96,0.943836734693878)
(0.97,0.957908629441624)
(0.98,0.971959595959596)
(0.99,0.985989949748744)
(1,1)
}
node [pos=0.5,pin={[pin edge={blue}, pin distance=0.25cm] 182:{NVT($B$)}}] {};
\end{axis}
\end{tikzpicture}
\end{subfigure}
\end{figure}

Analytical results are depicted on Figure \ref{Fig1} where the share ofparty $A$'s votes and seats are plotted in various scenarios according to some given $\alpha$. The curves represent the proportional share, and the six cases given by three transfer formulas and two gerrymandering parties. For example, PVT($A$) means that party $A$ draws the boundaries of constituencies and formula PVT is applied.

$\alpha = 0.3$ on Figure \ref{Fig1a}, which is a relatively low value. If party $A$ has a gerrymandering power, it's preference order is DVT $\succ$ NVT $\succ$ PVT according to Proposition \ref{Prop1}. PVT is worse than the proportional rule when $x > 0.6$.
If the other party $B$ has a gerrymandering power, party $A$ has the preference order of NVT $\succ$ DVT $\succ$ DVT according to Proposition \ref{Prop2}. NVT is better than the proportional rule when $x > 0.75$ despite the unfavorable constituency map. It can also be seen that NVT($B$), then PVT($B$), and finally, DVT($B$) exceeds NVT($A$) as $x$ becomes greater.

The curves are much closer to each other on Figure \ref{Fig1a} than on Figure \ref{Fig1b} with $\alpha = 0.6$.
Here the picture is more clear: party $A$ is over-represented if it has a gerrymandering power but under-represented if the other party draws district boundaries. Electoral systems used in practice are closer to Figure \ref{Fig1b} as single-member constituencies usually form the dominant part of the system.

Naturally, it is almost impossible that a party has such an extreme gerrymandering power in practice. However, the analysis of extreme cases show some tendencies: DVT may result in the most divergent outcomes, and NVT may be always better than PVT for the dominant party.

\subsection{Simulation results} \label{Sec32}

In the general case, when no party gerrymanders district boundaries, there is no hope for analytical results even if there is only two parties as the number of constituencies won depend on the territorial distribution of votes.
In order to get some insight we have carried out an investigation by simulations. It is assumed that there are $100$ constituencies and $x$ (the vote share of party $A$) has a continuous uniform distribution governed by two parameters: the expected value $k$ and the length $2h$, that is, it has a minimum of $k-h$ and a maximum of $k+h$.
We have made ten thousand runs in each simulation.

\begin{table}[!ht]
  \centering
  \caption{Number of runs when party $A$ obtains majority}
  \label{Table1}
    \begin{tabularx}{\textwidth}{l CCCCC} \toprule
    Value of $k$     & 0.51  & 0.52  & 0.53  & 0.54  & 0.55 \\ \midrule
    Under DVT, PVT and NVT & 7 762  & 9 344  & 9 890  & 9 980  & 9 999 \\
    Under DVT and PVT & 3     & 0     & 0     & 0     & 0 \\
    Under DVT and NVT & 0     & 0     & 0     & 0     & 0 \\
    Under PVT and NVT & 111   & 104   & 28    & 10    & 0 \\
    Under DVT & 1     & 0     & 0     & 0     & 0 \\
    Under PVT & 0     & 0     & 0     & 0     & 0 \\
    Under NVT & 100   & 86    & 22    & 3     & 1 \\
    Under no transfer formulas  & 2 023  & 466   & 60    & 7     & 0 \\ \bottomrule
    \end{tabularx}
\end{table}

The case $\alpha = 0.5$ and $h = 0.15$ for various values of $k$ is detailed in Table \ref{Table1}. The columns show how many runs have resulted in a majority for (the dominating) party $A$ according to all combination of the three vote transfer mechanisms (so their sum is ten thousand). PVT is not worth to consider since both DVT and NVT dominates it. NVT seems to be the best choice for party $A$, it leads to a win in significantly more cases than DVT. The reason is obvious: for a strong party with more than half of the votes in expected value, the list win is almost provided. However, local fluctuations may lead to a loss of some constituencies and in this way to a loss of the whole election. Then it should happen despite a robust win in some constituencies, from which party $A$ can benefit if formula NVT is applied.

\begin{figure}[htbp]
\centering
\caption{Comparison of DVT and NVT by simulation}
\label{Fig2}

\begin{subfigure}{0.48\textwidth}
	\centering
	\caption{$\alpha = 0.5$; $h = 0.15$}
	\label{Fig2a}    
\begin{tikzpicture}
\begin{axis}[width=0.95\textwidth, 
height=0.75\textwidth,
xlabel={Value of $k$ (expected value)},
axis y line*=left, 
ymin=0,
ymax=280,
ybar,
bar width = 6pt,
]

\addplot [color=red,fill] coordinates{
(0.51,4)
(0.52,0)
(0.53,0)
(0.54,0)
(0.55,0)
};

\addplot [color=blue,fill] coordinates{
(0.51,211)
(0.52,190)
(0.53,50)
(0.54,13)
(0.55,1)
};
\end{axis}
    
\begin{axis}[width=0.95\textwidth, 
height=0.75\textwidth,
axis x line*=none,
axis y line*=right,
hide x axis,
ymin=0.495,
ymax=0.645,
]
    
\addplot [color=red, very thick] coordinates { 
(0.51,0.521507719511076)
(0.52,0.543458474056405)
(0.53,0.565258851844435)
(0.54,0.586929987830437)
(0.55,0.607977009787712)
};

\addplot [color=blue, very thick] coordinates { 
(0.51,0.519135351581145)
(0.52,0.538728957385145)
(0.53,0.558113131109244)
(0.54,0.577443395461685)
(0.55,0.596200345098264)
};
\end{axis}
\end{tikzpicture}
\end{subfigure}
\begin{subfigure}{0.48\textwidth}
	\centering
	\caption{$\alpha = 0.5$; $h = 0.25$}
	\label{Fig2b}    
\begin{tikzpicture}
\begin{axis}[width=0.95\textwidth, 
height=0.75\textwidth,
xlabel={Value of $k$ (expected value)},
axis y line*=left, 
ymin=0,
ymax=280,
ybar,
bar width = 6pt,
]

\addplot [color=red,fill] coordinates{
(0.51,32)
(0.52,6)
(0.53,0)
(0.54,0)
(0.55,0)
(0.56,0)
(0.57,0)
(0.58,0)
};

\addplot [color=blue,fill] coordinates{
(0.51,239)
(0.52,270)
(0.53,223)
(0.54,109)
(0.55,37)
(0.56,13)
(0.57,6)
(0.58,1)
};
\end{axis}
    
\begin{axis}[width=0.95\textwidth, 
height=0.75\textwidth,
axis x line*=none,
axis y line*=right,
hide x axis,
ymin=0.495,
ymax=0.645,
]
    
\addplot [color=red, very thick] coordinates { 
(0.51,0.515194389167173)
(0.52,0.529966510109367)
(0.53,0.545419431582978)
(0.54,0.559694984086483)
(0.55,0.575047303893167)
(0.56,0.590178441696361)
(0.57,0.605060210065755)
(0.58,0.619646146880532)
};

\addplot [color=blue, very thick] coordinates { 
(0.51,0.514769589806261)
(0.52,0.529210685144854)
(0.53,0.544170270805085)
(0.54,0.558156608753017)
(0.55,0.573001760260099)
(0.56,0.587699363031428)
(0.57,0.602184431784387)
(0.58,0.61631739270502)
};
\end{axis}
\end{tikzpicture}
\end{subfigure}
\vspace{0.15cm} 
\rule{0ex}{10ex}
\begin{subfigure}{0.48\textwidth}
	\centering
	\caption{$\alpha = 0.7$; $h = 0.15$}
	\label{Fig2c}    
\begin{tikzpicture}
\begin{axis}[width=0.95\textwidth, 
height=0.75\textwidth,
xlabel={Value of $k$ (expected value)},
axis y line*=left, 
ymin=0,
ymax=140,
ybar,
bar width = 6pt,
]

\addplot [color=red,fill] coordinates{
(0.51,0)
(0.52,0)
(0.53,0)
(0.54,0)
};

\addplot [color=blue,fill] coordinates{
(0.51,16)
(0.52,34)
(0.53,42)
(0.54,9)
};
\end{axis}
    
\begin{axis}[width=0.95\textwidth, 
height=0.75\textwidth,
axis x line*=none,
axis y line*=right,
hide x axis,
ymin=0.495,
ymax=0.645,
]
    
\addplot [color=red, very thick] coordinates { 
(0.51,0.526140519883258)
(0.52,0.552113545654132)
(0.53,0.578728810909075)
(0.54,0.605483924679401)
};

\addplot [color=blue, very thick] coordinates { 
(0.51,0.524747459428755)
(0.52,0.549325146858464)
(0.53,0.574491999741061)
(0.54,0.599797304628761)
};
\end{axis}
\end{tikzpicture}
\end{subfigure}
\begin{subfigure}{0.48\textwidth}
	\centering
	\caption{$\alpha = 0.7$; $h = 0.25$}
	\label{Fig2d}    
\begin{tikzpicture}
\begin{axis}[width=0.95\textwidth, 
height=0.75\textwidth,
xlabel={Value of $k$ (expected value)},
axis y line*=left, 
ymin=0,
ymax=140,
ybar,
bar width = 6pt,
]

\addplot [color=red,fill] coordinates{
(0.51,2)
(0.52,0)
(0.53,0)
(0.54,0)
(0.55,0)
(0.56,0)
(0.57,0)
(0.58,0)
};

\addplot [color=blue,fill] coordinates{
(0.51,48)
(0.52,89)
(0.53,124)
(0.54,89)
(0.55,33)
(0.56,16)
(0.57,2)
(0.58,2)
};
\end{axis}
    
\begin{axis}[width=0.95\textwidth, 
height=0.75\textwidth,
axis x line*=none,
axis y line*=right,
hide x axis,
ymin=0.495,
ymax=0.645,
]
    
\addplot [color=red, very thick] coordinates { 
(0.51,0.516736896881601)
(0.52,0.534112182384335)
(0.53,0.550400478377325)
(0.54,0.568629689916328)
(0.55,0.585614674155906)
(0.56,0.602024575786073)
(0.57,0.619215748582437)
(0.58,0.636350846940165)
};

\addplot [color=blue, very thick] coordinates { 
(0.51,0.516550392107346)
(0.52,0.533633681401821)
(0.53,0.549743742229298)
(0.54,0.567629841622269)
(0.55,0.584369422382675)
(0.56,0.600524889463733)
(0.57,0.617448235521735)
(0.58,0.634284938047674)
};
\end{axis}
\end{tikzpicture}
\end{subfigure}

\vspace{0.1cm} 

\begin{tikzpicture}
        \begin{customlegend}[legend columns=1,legend entries={Average vote share of party $A$ under DVT (right scale),Average vote share of party $A$ under NVT (right scale)}]
        \addlegendimage{color=red, very thick}
        \addlegendimage{color=blue, very thick}   
        \end{customlegend}
\end{tikzpicture}

\pgfplotsset{
legend image code/.code={%
\draw[#1,fill] (0cm,-0.1cm) rectangle (0.125cm,0.25cm) (0.25cm,-0.1cm) rectangle (0.375cm,0.175cm);
},
}
\vspace{0.1cm}
\begin{tikzpicture}        
        \begin{customlegend}[legend columns=1,legend entries={Number of runs: DVT results in a majority contrary to NVT (left scale),Number of runs: NVT results in a majority contrary to DVT (left scale)}]
        \addlegendimage{color=red, very thick}
        \addlegendimage{color=blue, very thick}   
        \end{customlegend}
\end{tikzpicture}
\end{figure}
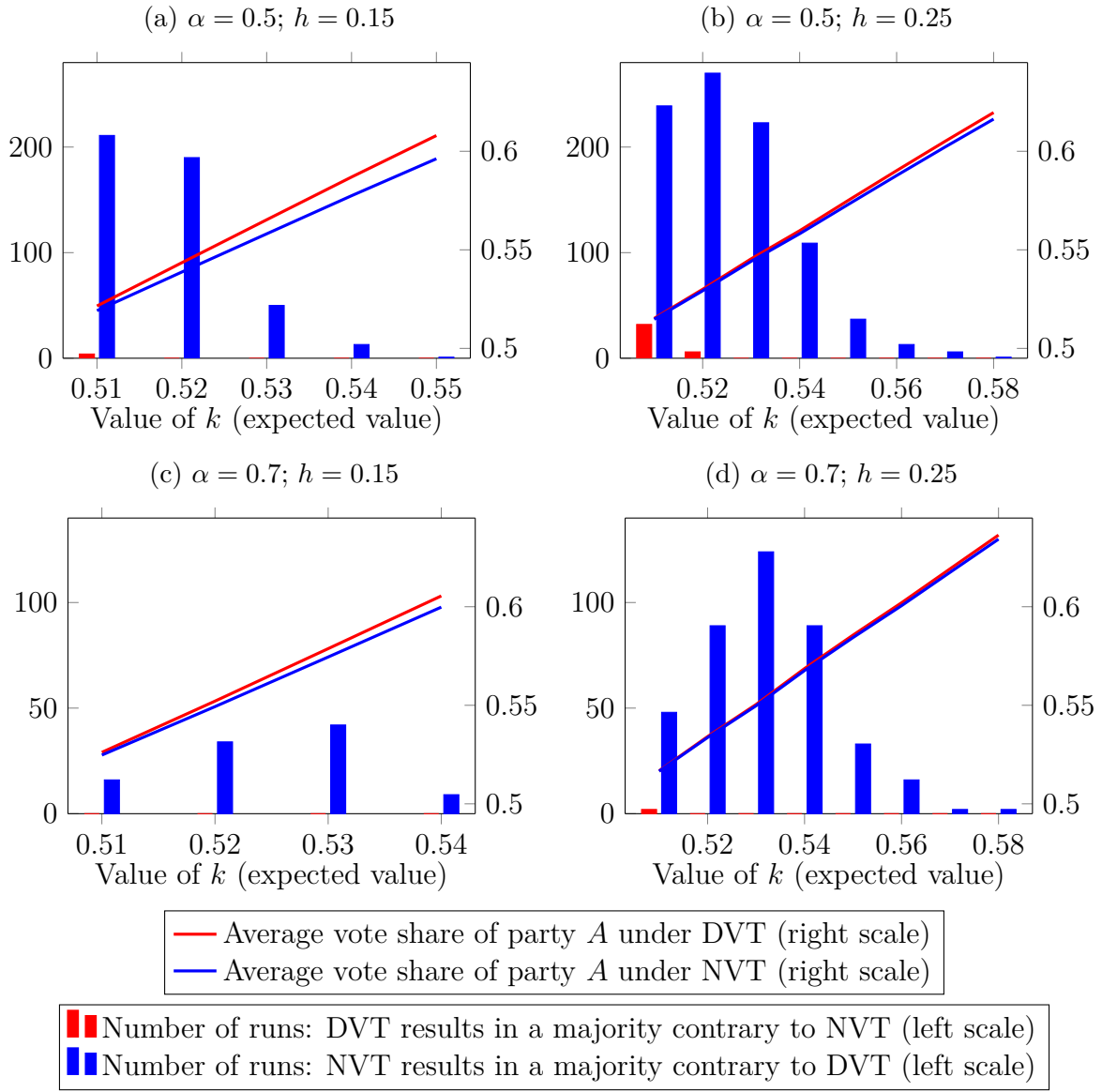

One may think that the dominating party has still some incentives to raise the share of constituency seats in order to increase its majority. It is actually true as Figure \ref{Fig2} presents. Here the lines show the average vote share of the party $A$ in the simulation, while bars show the number of runs where DVT vs. NVT result in a majority contrary to the other formula (from the total ten thousand runs). In the terminology of Table \ref{Table1}, the latter means Under DVT and PVT plus Under DVT vs. Under PVT and NTV plus Under NVT.

On the basis of Figure \ref{Fig2}, DVT always result in a larger majority \emph{ex ante}, although NVT provides a win more often \emph{ex post}. Difference between the averages declines by increasing the ratio of constituency seats ($\alpha$) as well as by a rise in variance ($h$). DVT is better than NVT only in a few cases with respect to a majority. NVT clearly overcomes DVT especially if the expected value of party $A$'s vote share is only a bit above half but not too close to it.

It seems there exists a kind of trade-off between DVT and NVT: the first will probably give more seats to the dominating party but it endangers its majority. Therefore a risk-averse party prefers rule NVT to DVT.
All results are robust with respect to the share of list seats, at least if the dominating party's votes are distributed uniformly with an arbitrary variation among the constituencies.

\subsection{Extension to more parties} \label{Sec33}

In this case the vote shares of parties depend on their relative strength in each constituency and cannot be governed by just one parameter as for a two-party system. For example, it can occur that the third party $C$ has no chance to win a constituency, but it may also be a strong regional party which is a clear favorite in some districts.
It makes the derivation of analytical results practically impossible. Simulations are difficult to carry out, too, since there should be introduced at least two independent random variables in each constituency and there can be found a lot of possibilities to determine their correlation. Therefore this topic is left for future research.

\subsection{Electoral strategies under different vote transfer formulas} \label{Sec34}

All mixed electoral systems are subject to strategic behavior \citep{BochslerBernauer2014, Bochsler2015}. We consider three possible ways of manipulation: party-splitting, vote transferring and stronghold-splitting.

Party-splitting refers to a paradox when large parties can gain by splitting their votes on several lists. It is a relevant issue in most mixed-member electoral systems since a party winning a number of constituencies is usually punished in the allocation of compensation seats. It does not emerge in our model because voters have only one vote.

It leads to the second possible strategy, vote transferring. Then parties avoid to win a district mandate and transfer their votes to the proportional part of the system instead.\footnote{~It seems to be difficult not to win a district consciously. However, it can be regarded not as a deliberate strategy but a kind of monotonicity principle: it will be strange if a party gains by getting fewer votes.}
It only pays out if the 'price' of compensatory mandates is smaller than at least one district mandate, which is not probable in the model until the share of list seats is very high (small $\alpha$).\footnote{~An example can be seen in Section \ref{Sec31}, where party $A$ with a gerrymandering power loses under PVT compared to the proportional rule if $\alpha < 1/2 x$.} In fact, its emergence requires that $\alpha < 0.5$ since the number of votes cast in each constituency is equal and the number of list votes cannot be smaller than their sum.

Note also that the \emph{ex post} variation (i.e. the number of votes cast count, not the number of registered voters) of the size of constituencies increases the probability of a successful vote-transfer strategy. However, it remains more difficult under NVT than under PVT since the former results in more list votes. Moreover, DVT entirely excludes the exploitation of this strategy as there exists no possible gain from losing a district mandate.

The third strategy is stronghold-splitting: in a party stronghold, its candidate can win with a large margin over the second-largest party but excess votes does not count under DVT and PVT. So the party can obtain the mandate with fewer votes by presenting two candidates, one of them running on a 'clone list', that is, votes cast on its losing candidate are added to the list votes. The system cannot be manipulated under DVT and NVT in this way.

To conclude, vote transfer electoral systems perform well with respect to strategic manipulation. Split-ticket voting is impossible, stronghold-splitting occurs only in the case of PVT, and DVT is immune still to a vote-transfer strategy. Among the three formulas, DVT is the most robust, while PVT is the most vulnerable from the viewpoint of manipulation.

\section{The model and Hungarian electoral system} \label{Sec4}

Since the new Constitution of 2012 the electoral system of Hungary consists of two tiers. Each elector residing in the country casts two votes, one for an individual candidate and one for a party list. Voters having no domicile in Hungary may vote for a party list.

The candidate vote determines the allocation of $106$ mandates from single-member districts in the first tier. Constituencies have a territorial base and approximately equal size with respect to the number of registered voters. They are won by the candidate who got the most votes.

Further $93$ mandates are allocated among the parties with a national list in proportion to the votes cast for their list and the surplus votes according to the d'Hondt formula.\footnote{~Minorities are able to set a minority list and obtain the first mandate for only a quarter of the votes required for a 'normal' party mandate. One must to register as a member of national minority if wishes to vote for national minority lists. In that case, this voter is not allowed to vote for a party list. In the 2014 elections no national minority obtained a parliamentary seat.}

Surplus votes comprise the votes cast for party candidates in the single-member constituencies who lost and the votes cast for party candidates in the single-member constituencies which have effectively not been needed to obtain seats (the difference in the number of votes between on the most successful and second candidates in a given constituency). There exists a $5$\% threshold in case of a party-list, $10$\% in case of two parties' joint list, and $15$\% in case of a joint lists of three or more parties.\footnote{~A detailed description of the electoral rules can be found at \url{http://www.ipu.org/parline-e/reports/2141_B.htm}.}

It closely follows the model presented in Section~\ref{Sec2}, however, voters have two votes instead of one. Candidate and list votes were not distinguished in order to eliminate vote splitting, however, it is not frequent in Hungary as Table \ref{Table2} illustrates by the data of 2014 parliamentary elections.
Note that only inland list votes should be taken into account since citizens without domicile in the country do not vote in single-member districts. Larger parties\footnote{~In the following, we call party lists registered for the 2014 election to parties, but some of them are essentially an alliance of parties.} get a bit more, smaller parties get a bit fewer votes in constituencies, however, the difference is not significant.

\begin{table}[!ht]
  \centering
  \caption{Aggregated result, 2014 Hungarian parliamentary election}
  \label{Table2}
    \begin{tabularx}{1\textwidth}{l RRR} \toprule
    Party & \multicolumn{1}{C}{Inland list \linebreak votes (L)} & \multicolumn{1}{C}{Candidate votes (C)} & \multicolumn{1}{C}{Difference (1-C/L)} \\ \midrule
    FIDESZ-KDNP & 2 142 142 & 2 165 342 & -1.07\% \\
    MSZP-EGYÜTT-DK-PM-MLP & 1 289 311 & 1 317 879 & -2.17\% \\
    JOBBIK & 1 017 550 & 1 000 637 & 1.69\% \\
    LMP   & 268 840 & 244 191 & 10.09\% \\ \bottomrule
    \end{tabularx}
\end{table}

The share of constituency seats is $\alpha = 106/199 \approx 53.27$\%. The new election law specifies the use of NVT in order to calculate list votes.

\section{Application to the 2014 Hungarian parliamentary election} \label{Sec5}

The last Hungarian parliamentary election took place on 6 April 2014 and so far it is the only election organized according to the new law.

In parliamentary elections between 1990 and 2010 a different system was applied, investigated by \citet{Benoit2001} and \citet{BenoitSchiemann2001}, among others. It was a much more complex system with three tiers. Briefly, there was $386$ representatives from $176$ single-member constituencies and $210$ candidates from party lists (a maximum of $152$ on territorial and a minimum of $58$ on national party lists). It means $\alpha = 176/386 \approx 46$\%. As each territorial list as well as the national list had its own divisor for converting votes into seats, this system suffered from the population paradox, namely, some parties might lose by getting more votes or by the opposition obtaining fewer votes \citep{Tasnadi2008}.

The electoral reform of 2012 has brought substantial changes including a redrawing of single-member constituencies and an essential simplification of the proportional tier of the system. A seemingly minor change was a modification in the calculation of surplus votes: in the elections between 1990 and 2010 only the votes in constituencies cast for party candidates who lost were transferred into list votes, corresponding to our formula PVT.

In the following hypothetical results will be derived on the basis of 2014 results. The focus is on two parameters changed by the electoral reform: the share of constituency seats ($\alpha$) and the formula applied for calculating surplus votes.

Party FIDESZ-KDNP has won $96$ constituencies, while MSZP-EGYÜTT-DK-PM-MLP has obtained the remaining $10$.
Four parties have got more votes than the appropriate threshold for list votes, they are FIDESZ-KDNP, MSZP-EGYÜTT-DK-PM-MLP, JOBBIK and LMP.
Number of votes for these party lists are given in Table \ref{Table3} according to the three vote transfer formulas. The last column corresponds to the official result.

DVT does not contain any corrections. PVT favors JOBBIK and LMP, the parties which have not won any constituency. They almost double the number of effective votes (small differences are due to vote splitting). Contrarily, PVT is not a favorable procedure for FIDESZ-KDNP since it has few losing candidates. NVT is equivalent to PVT in the case of JOBBIK and LMP, adds some votes to MSZP-EGYÜTT-DK-PM-MLP but FIDESZ-KDNP benefits most from its use.

\begin{table}[!ht]
  \centering
  \caption{Number of list votes, 2014 Hungarian parliamentary election}
  \label{Table3}
    \begin{tabularx}{\textwidth}{l RRR} \toprule
    Party & \multicolumn{1}{C}{DVT} & \multicolumn{1}{C}{PVT} & \multicolumn{1}{C}{NVT} \\ \midrule
    FIDESZ-KDNP & 2 264 780 & 2 440 963 & 3 205 661 \\
    MSZP-EGYÜTT-DK-PM-MLP & 1 290 806 & 2 410 128 & 2 432 492 \\
    JOBBIK & 1 020 476 & 2 021 113 & 2 021 113 \\
    LMP   & 269 414 & 513 605 & 513 605 \\ \bottomrule
    \end{tabularx} 
\end{table}

Parameter $\alpha$ can be varied even if the results in the constituencies and the votes cast is fixed as given in Table \ref{Table3}. We have chosen the number of list seats to be between $50$ and $150$, namely, $67.95\% \approx 106/156 \geq \alpha \geq 106/256 \approx 41.41\%$. The number of mandates for each party are presented on Figure \ref{Fig3}.

\begin{figure}[htbp]
\centering
\caption{Mandate share of parties, 2014 Hungarian parliamentary election}
\label{Fig3}

\begin{subfigure}{\textwidth}
	\centering
	\caption{FIDESZ-KDNP}
	\label{Fig3a}
	
\begin{tikzpicture}
\begin{axis}[width=0.98\textwidth, 
height=0.65\textwidth,
xmin=47.5,
xmax=152.5,
xlabel = Number of list mandates,
ylabel = Mandate share,
ymajorgrids,
legend entries={DVT$\quad$,PVT$\quad$,NVT},
legend style={at={(0.5,-0.15)},anchor = north,legend columns = 3},
scaled ticks=false,
yticklabel=\pgfmathparse{100*\tick}\pgfmathprintnumber{\pgfmathresult}\,\%,
yticklabel style={/pgf/number format/.cd,fixed,precision=2},
]

\addplot[black,dotted,very thick] coordinates {
(50,0.769230769230769)
(51,0.764331210191083)
(52,0.765822784810127)
(53,0.761006289308176)
(54,0.7625)
(55,0.757763975155279)
(56,0.753086419753086)
(57,0.754601226993865)
(58,0.75609756097561)
(59,0.751515151515151)
(60,0.746987951807229)
(61,0.748502994011976)
(62,0.744047619047619)
(63,0.745562130177515)
(64,0.747058823529412)
(65,0.742690058479532)
(66,0.738372093023256)
(67,0.739884393063584)
(68,0.741379310344828)
(69,0.737142857142857)
(70,0.732954545454545)
(71,0.728813559322034)
(72,0.730337078651685)
(73,0.731843575418994)
(74,0.727777777777778)
(75,0.723756906077348)
(76,0.725274725274725)
(77,0.721311475409836)
(78,0.722826086956522)
(79,0.718918918918919)
(80,0.720430107526882)
(81,0.716577540106952)
(82,0.718085106382979)
(83,0.714285714285714)
(84,0.715789473684211)
(85,0.712041884816754)
(86,0.713541666666667)
(87,0.715025906735751)
(88,0.711340206185567)
(89,0.707692307692308)
(90,0.704081632653061)
(91,0.705583756345178)
(92,0.702020202020202)
(93,0.703517587939699)
(94,0.7)
(95,0.701492537313433)
(96,0.698019801980198)
(97,0.699507389162562)
(98,0.696078431372549)
(99,0.697560975609756)
(100,0.694174757281553)
(101,0.695652173913043)
(102,0.692307692307692)
(103,0.69377990430622)
(104,0.69047619047619)
(105,0.691943127962085)
(106,0.688679245283019)
(107,0.685446009389671)
(108,0.686915887850467)
(109,0.683720930232558)
(110,0.685185185185185)
(111,0.682027649769585)
(112,0.68348623853211)
(113,0.680365296803653)
(114,0.681818181818182)
(115,0.678733031674208)
(116,0.68018018018018)
(117,0.677130044843049)
(118,0.678571428571429)
(119,0.675555555555556)
(120,0.676991150442478)
(121,0.674008810572687)
(122,0.671052631578947)
(123,0.672489082969432)
(124,0.669565217391304)
(125,0.670995670995671)
(126,0.668103448275862)
(127,0.665236051502146)
(128,0.666666666666667)
(129,0.668085106382979)
(130,0.665254237288136)
(131,0.666666666666667)
(132,0.663865546218487)
(133,0.665271966527197)
(134,0.6625)
(135,0.663900414937759)
(136,0.661157024793388)
(137,0.65843621399177)
(138,0.659836065573771)
(139,0.661224489795918)
(140,0.658536585365854)
(141,0.65587044534413)
(142,0.657258064516129)
(143,0.654618473895582)
(144,0.656)
(145,0.653386454183267)
(146,0.650793650793651)
(147,0.652173913043478)
(148,0.653543307086614)
(149,0.650980392156863)
(150,0.65234375)
};

\addplot[red,very thick] coordinates {
(50,0.724358974358974)
(51,0.719745222929936)
(52,0.721518987341772)
(53,0.716981132075472)
(54,0.7125)
(55,0.714285714285714)
(56,0.709876543209877)
(57,0.705521472392638)
(58,0.701219512195122)
(59,0.703030303030303)
(60,0.698795180722892)
(61,0.694610778443114)
(62,0.696428571428571)
(63,0.692307692307692)
(64,0.688235294117647)
(65,0.690058479532164)
(66,0.686046511627907)
(67,0.682080924855491)
(68,0.683908045977011)
(69,0.68)
(70,0.676136363636364)
(71,0.677966101694915)
(72,0.674157303370787)
(73,0.670391061452514)
(74,0.672222222222222)
(75,0.668508287292818)
(76,0.664835164835165)
(77,0.666666666666667)
(78,0.66304347826087)
(79,0.65945945945946)
(80,0.661290322580645)
(81,0.657754010695187)
(82,0.654255319148936)
(83,0.656084656084656)
(84,0.652631578947368)
(85,0.649214659685864)
(86,0.645833333333333)
(87,0.647668393782383)
(88,0.644329896907217)
(89,0.646153846153846)
(90,0.642857142857143)
(91,0.639593908629442)
(92,0.641414141414141)
(93,0.638190954773869)
(94,0.635)
(95,0.63681592039801)
(96,0.633663366336634)
(97,0.630541871921182)
(98,0.632352941176471)
(99,0.629268292682927)
(100,0.62621359223301)
(101,0.623188405797101)
(102,0.625)
(103,0.62200956937799)
(104,0.623809523809524)
(105,0.62085308056872)
(106,0.617924528301887)
(107,0.619718309859155)
(108,0.616822429906542)
(109,0.613953488372093)
(110,0.615740740740741)
(111,0.612903225806452)
(112,0.610091743119266)
(113,0.611872146118721)
(114,0.609090909090909)
(115,0.606334841628959)
(116,0.603603603603604)
(117,0.605381165919282)
(118,0.602678571428571)
(119,0.6)
(120,0.601769911504425)
(121,0.599118942731278)
(122,0.600877192982456)
(123,0.59825327510917)
(124,0.595652173913043)
(125,0.597402597402597)
(126,0.594827586206897)
(127,0.592274678111588)
(128,0.58974358974359)
(129,0.591489361702128)
(130,0.588983050847458)
(131,0.586497890295359)
(132,0.588235294117647)
(133,0.585774058577406)
(134,0.583333333333333)
(135,0.5850622406639)
(136,0.582644628099174)
(137,0.580246913580247)
(138,0.581967213114754)
(139,0.579591836734694)
(140,0.58130081300813)
(141,0.578947368421053)
(142,0.576612903225806)
(143,0.57429718875502)
(144,0.576)
(145,0.573705179282869)
(146,0.571428571428571)
(147,0.573122529644269)
(148,0.570866141732283)
(149,0.568627450980392)
(150,0.5703125)
};

\addplot[blue,dashed,very thick] coordinates {
(50,0.743589743589744)
(51,0.738853503184713)
(52,0.740506329113924)
(53,0.735849056603774)
(54,0.7375)
(55,0.732919254658385)
(56,0.728395061728395)
(57,0.730061349693252)
(58,0.725609756097561)
(59,0.721212121212121)
(60,0.72289156626506)
(61,0.718562874251497)
(62,0.720238095238095)
(63,0.715976331360947)
(64,0.711764705882353)
(65,0.713450292397661)
(66,0.709302325581395)
(67,0.705202312138728)
(68,0.706896551724138)
(69,0.702857142857143)
(70,0.704545454545455)
(71,0.700564971751412)
(72,0.696629213483146)
(73,0.698324022346369)
(74,0.7)
(75,0.696132596685083)
(76,0.692307692307692)
(77,0.693989071038251)
(78,0.690217391304348)
(79,0.686486486486486)
(80,0.682795698924731)
(81,0.684491978609626)
(82,0.680851063829787)
(83,0.682539682539683)
(84,0.678947368421053)
(85,0.680628272251309)
(86,0.677083333333333)
(87,0.673575129533679)
(88,0.675257731958763)
(89,0.671794871794872)
(90,0.673469387755102)
(91,0.67005076142132)
(92,0.666666666666667)
(93,0.668341708542714)
(94,0.665)
(95,0.666666666666667)
(96,0.663366336633663)
(97,0.660098522167488)
(98,0.661764705882353)
(99,0.658536585365854)
(100,0.655339805825243)
(101,0.657004830917874)
(102,0.653846153846154)
(103,0.655502392344498)
(104,0.652380952380952)
(105,0.654028436018957)
(106,0.650943396226415)
(107,0.647887323943662)
(108,0.649532710280374)
(109,0.646511627906977)
(110,0.643518518518518)
(111,0.645161290322581)
(112,0.642201834862385)
(113,0.639269406392694)
(114,0.640909090909091)
(115,0.638009049773756)
(116,0.63963963963964)
(117,0.63677130044843)
(118,0.638392857142857)
(119,0.635555555555556)
(120,0.632743362831858)
(121,0.634361233480176)
(122,0.631578947368421)
(123,0.633187772925764)
(124,0.630434782608696)
(125,0.627705627705628)
(126,0.629310344827586)
(127,0.626609442060086)
(128,0.623931623931624)
(129,0.625531914893617)
(130,0.622881355932203)
(131,0.624472573839662)
(132,0.621848739495798)
(133,0.619246861924686)
(134,0.620833333333333)
(135,0.618257261410788)
(136,0.619834710743802)
(137,0.617283950617284)
(138,0.618852459016393)
(139,0.616326530612245)
(140,0.613821138211382)
(141,0.615384615384615)
(142,0.612903225806452)
(143,0.610441767068273)
(144,0.612)
(145,0.609561752988048)
(146,0.607142857142857)
(147,0.608695652173913)
(148,0.606299212598425)
(149,0.607843137254902)
(150,0.60546875)
};

\addplot[brown,dashdotted,very thick] coordinates {
(47.5,2/3)
(152.5,2/3)
};
\draw[green,ultra thick] ({axis cs:93,0}|-{rel axis cs:0,0}) -- ({axis cs:93,0}|-{rel axis cs:0,1});
\end{axis}
\end{tikzpicture}
\end{subfigure}
\vspace{0.5cm}

\begin{subfigure}{\textwidth}
	\centering
	\caption{MSZP-EGYÜTT-DK-PM-MLP}
	\label{Fig3b}
	
\begin{tikzpicture}
\begin{axis}[width=0.98\textwidth, 
height=0.65\textwidth,
xmin=47.5,
xmax=152.5,
xlabel = Number of list mandates,
ylabel = Mandate share,
ymajorgrids,
legend entries={DVT$\quad$,PVT$\quad$,NVT},
legend style={at={(0.5,-0.15)},anchor = north,legend columns = 3},
scaled ticks=false,
yticklabel=\pgfmathparse{100*\tick}\pgfmathprintnumber{\pgfmathresult}\,\%,
yticklabel style={/pgf/number format/.cd,fixed,precision=2}
]

\addplot[black,dotted,very thick] coordinates {
(50,0.147435897435897)
(51,0.152866242038217)
(52,0.151898734177215)
(53,0.150943396226415)
(54,0.15)
(55,0.15527950310559)
(56,0.154320987654321)
(57,0.153374233128834)
(58,0.152439024390244)
(59,0.157575757575758)
(60,0.156626506024096)
(61,0.155688622754491)
(62,0.160714285714286)
(63,0.159763313609467)
(64,0.158823529411765)
(65,0.157894736842105)
(66,0.162790697674419)
(67,0.161849710982659)
(68,0.160919540229885)
(69,0.16)
(70,0.164772727272727)
(71,0.163841807909605)
(72,0.162921348314607)
(73,0.162011173184358)
(74,0.166666666666667)
(75,0.165745856353591)
(76,0.164835164835165)
(77,0.169398907103825)
(78,0.168478260869565)
(79,0.167567567567568)
(80,0.166666666666667)
(81,0.171122994652406)
(82,0.170212765957447)
(83,0.169312169312169)
(84,0.168421052631579)
(85,0.172774869109948)
(86,0.171875)
(87,0.170984455958549)
(88,0.170103092783505)
(89,0.174358974358974)
(90,0.173469387755102)
(91,0.17258883248731)
(92,0.176767676767677)
(93,0.175879396984925)
(94,0.175)
(95,0.174129353233831)
(96,0.178217821782178)
(97,0.177339901477833)
(98,0.176470588235294)
(99,0.175609756097561)
(100,0.179611650485437)
(101,0.178743961352657)
(102,0.177884615384615)
(103,0.177033492822967)
(104,0.180952380952381)
(105,0.180094786729858)
(106,0.179245283018868)
(107,0.183098591549296)
(108,0.182242990654206)
(109,0.181395348837209)
(110,0.180555555555556)
(111,0.184331797235023)
(112,0.18348623853211)
(113,0.182648401826484)
(114,0.181818181818182)
(115,0.18552036199095)
(116,0.184684684684685)
(117,0.183856502242152)
(118,0.183035714285714)
(119,0.186666666666667)
(120,0.185840707964602)
(121,0.185022026431718)
(122,0.18859649122807)
(123,0.187772925764192)
(124,0.18695652173913)
(125,0.186147186147186)
(126,0.189655172413793)
(127,0.188841201716738)
(128,0.188034188034188)
(129,0.187234042553191)
(130,0.190677966101695)
(131,0.189873417721519)
(132,0.189075630252101)
(133,0.188284518828452)
(134,0.191666666666667)
(135,0.190871369294606)
(136,0.190082644628099)
(137,0.193415637860082)
(138,0.192622950819672)
(139,0.191836734693878)
(140,0.191056910569106)
(141,0.194331983805668)
(142,0.193548387096774)
(143,0.192771084337349)
(144,0.192)
(145,0.195219123505976)
(146,0.194444444444444)
(147,0.193675889328063)
(148,0.192913385826772)
(149,0.196078431372549)
(150,0.1953125)
};

\addplot[red,very thick] coordinates {
(50,0.166666666666667)
(51,0.171974522292994)
(52,0.170886075949367)
(53,0.169811320754717)
(54,0.175)
(55,0.173913043478261)
(56,0.172839506172839)
(57,0.177914110429448)
(58,0.176829268292683)
(59,0.175757575757576)
(60,0.180722891566265)
(61,0.179640718562874)
(62,0.178571428571429)
(63,0.183431952662722)
(64,0.182352941176471)
(65,0.181286549707602)
(66,0.186046511627907)
(67,0.184971098265896)
(68,0.183908045977011)
(69,0.188571428571429)
(70,0.1875)
(71,0.186440677966102)
(72,0.185393258426966)
(73,0.189944134078212)
(74,0.188888888888889)
(75,0.193370165745856)
(76,0.192307692307692)
(77,0.191256830601093)
(78,0.195652173913043)
(79,0.194594594594595)
(80,0.193548387096774)
(81,0.197860962566845)
(82,0.196808510638298)
(83,0.195767195767196)
(84,0.2)
(85,0.198952879581152)
(86,0.197916666666667)
(87,0.196891191709845)
(88,0.201030927835052)
(89,0.2)
(90,0.198979591836735)
(91,0.203045685279188)
(92,0.202020202020202)
(93,0.206030150753769)
(94,0.205)
(95,0.203980099502488)
(96,0.207920792079208)
(97,0.206896551724138)
(98,0.205882352941176)
(99,0.204878048780488)
(100,0.20873786407767)
(101,0.207729468599034)
(102,0.206730769230769)
(103,0.210526315789474)
(104,0.20952380952381)
(105,0.208530805687204)
(106,0.212264150943396)
(107,0.211267605633803)
(108,0.210280373831776)
(109,0.213953488372093)
(110,0.212962962962963)
(111,0.211981566820276)
(112,0.215596330275229)
(113,0.214611872146119)
(114,0.213636363636364)
(115,0.217194570135747)
(116,0.216216216216216)
(117,0.2152466367713)
(118,0.21875)
(119,0.217777777777778)
(120,0.216814159292035)
(121,0.220264317180617)
(122,0.219298245614035)
(123,0.218340611353712)
(124,0.221739130434783)
(125,0.220779220779221)
(126,0.219827586206897)
(127,0.223175965665236)
(128,0.222222222222222)
(129,0.221276595744681)
(130,0.220338983050847)
(131,0.223628691983122)
(132,0.222689075630252)
(133,0.225941422594142)
(134,0.225)
(135,0.224066390041494)
(136,0.227272727272727)
(137,0.226337448559671)
(138,0.225409836065574)
(139,0.228571428571429)
(140,0.227642276422764)
(141,0.226720647773279)
(142,0.225806451612903)
(143,0.228915662650602)
(144,0.228)
(145,0.227091633466135)
(146,0.23015873015873)
(147,0.229249011857708)
(148,0.228346456692913)
(149,0.231372549019608)
(150,0.23046875)
};

\addplot[blue,dashed,very thick] coordinates {
(50,0.16025641025641)
(51,0.159235668789809)
(52,0.158227848101266)
(53,0.163522012578616)
(54,0.1625)
(55,0.161490683229814)
(56,0.166666666666667)
(57,0.165644171779141)
(58,0.170731707317073)
(59,0.16969696969697)
(60,0.168674698795181)
(61,0.167664670658683)
(62,0.166666666666667)
(63,0.171597633136095)
(64,0.170588235294118)
(65,0.169590643274854)
(66,0.174418604651163)
(67,0.173410404624277)
(68,0.172413793103448)
(69,0.177142857142857)
(70,0.176136363636364)
(71,0.175141242937853)
(72,0.179775280898876)
(73,0.17877094972067)
(74,0.177777777777778)
(75,0.176795580110497)
(76,0.181318681318681)
(77,0.180327868852459)
(78,0.179347826086957)
(79,0.183783783783784)
(80,0.182795698924731)
(81,0.181818181818182)
(82,0.186170212765957)
(83,0.185185185185185)
(84,0.184210526315789)
(85,0.183246073298429)
(86,0.1875)
(87,0.186528497409326)
(88,0.185567010309278)
(89,0.18974358974359)
(90,0.188775510204082)
(91,0.187817258883249)
(92,0.191919191919192)
(93,0.190954773869347)
(94,0.19)
(95,0.189054726368159)
(96,0.188118811881188)
(97,0.192118226600985)
(98,0.191176470588235)
(99,0.195121951219512)
(100,0.194174757281553)
(101,0.193236714975845)
(102,0.197115384615385)
(103,0.196172248803828)
(104,0.195238095238095)
(105,0.194312796208531)
(106,0.19811320754717)
(107,0.197183098591549)
(108,0.196261682242991)
(109,0.2)
(110,0.199074074074074)
(111,0.19815668202765)
(112,0.197247706422018)
(113,0.200913242009132)
(114,0.2)
(115,0.199095022624434)
(116,0.198198198198198)
(117,0.201793721973094)
(118,0.200892857142857)
(119,0.204444444444444)
(120,0.20353982300885)
(121,0.202643171806167)
(122,0.206140350877193)
(123,0.205240174672489)
(124,0.204347826086957)
(125,0.203463203463203)
(126,0.202586206896552)
(127,0.206008583690987)
(128,0.205128205128205)
(129,0.204255319148936)
(130,0.207627118644068)
(131,0.206751054852321)
(132,0.205882352941176)
(133,0.209205020920502)
(134,0.208333333333333)
(135,0.20746887966805)
(136,0.206611570247934)
(137,0.209876543209877)
(138,0.209016393442623)
(139,0.212244897959184)
(140,0.211382113821138)
(141,0.210526315789474)
(142,0.209677419354839)
(143,0.21285140562249)
(144,0.212)
(145,0.211155378486056)
(146,0.214285714285714)
(147,0.213438735177866)
(148,0.21259842519685)
(149,0.211764705882353)
(150,0.21484375)
};
\draw[green,ultra thick] ({axis cs:93,0}|-{rel axis cs:0,0}) -- ({axis cs:93,0}|-{rel axis cs:0,1});
\end{axis}
\end{tikzpicture}
\end{subfigure}
\end{figure}

\begin{figure}
\ContinuedFloat
\begin{subfigure}{\textwidth}
	\centering
	\caption{JOBBIK}
	\label{Fig3c}
	
\begin{tikzpicture}
\begin{axis}[width=0.98\textwidth, 
height=0.65\textwidth,
xmin=47.5,
xmax=152.5,
xlabel = Number of list mandates,
ylabel = Mandate share,
ymajorgrids,
legend entries={DVT$\quad$,PVT$\quad$,NVT},
legend style={at={(0.5,-0.15)},anchor = north,legend columns = 3},
scaled ticks=false,
yticklabel=\pgfmathparse{100*\tick}\pgfmathprintnumber{\pgfmathresult}\,\%,
yticklabel style={/pgf/number format/.cd,fixed,precision=2}
]

\addplot[black,dotted,very thick] coordinates {
(50,7.05128205128205E-02)
(51,7.00636942675159E-02)
(52,0.069620253164557)
(53,6.91823899371069E-02)
(54,0.06875)
(55,6.83229813664596E-02)
(56,7.40740740740741E-02)
(57,7.36196319018405E-02)
(58,7.31707317073171E-02)
(59,7.27272727272727E-02)
(60,7.83132530120482E-02)
(61,7.78443113772455E-02)
(62,7.73809523809524E-02)
(63,7.69230769230769E-02)
(64,7.64705882352941E-02)
(65,8.18713450292398E-02)
(66,8.13953488372093E-02)
(67,8.09248554913295E-02)
(68,8.04597701149425E-02)
(69,8.57142857142857E-02)
(70,8.52272727272727E-02)
(71,8.47457627118644E-02)
(72,8.42696629213483E-02)
(73,8.37988826815642E-02)
(74,8.33333333333333E-02)
(75,8.83977900552486E-02)
(76,8.79120879120879E-02)
(77,0.087431693989071)
(78,8.69565217391304E-02)
(79,9.18918918918919E-02)
(80,9.13978494623656E-02)
(81,9.09090909090909E-02)
(82,9.04255319148936E-02)
(83,9.52380952380952E-02)
(84,9.47368421052632E-02)
(85,9.42408376963351E-02)
(86,0.09375)
(87,9.32642487046632E-02)
(88,9.27835051546392E-02)
(89,9.23076923076923E-02)
(90,9.69387755102041E-02)
(91,9.64467005076142E-02)
(92,0.095959595959596)
(93,9.54773869346734E-02)
(94,0.1)
(95,9.95024875621891E-02)
(96,0.099009900990099)
(97,9.85221674876847E-02)
(98,0.102941176470588)
(99,0.102439024390244)
(100,0.101941747572816)
(101,0.101449275362319)
(102,0.105769230769231)
(103,0.105263157894737)
(104,0.104761904761905)
(105,0.104265402843602)
(106,0.10377358490566)
(107,0.103286384976526)
(108,0.102803738317757)
(109,0.106976744186047)
(110,0.106481481481481)
(111,0.105990783410138)
(112,0.105504587155963)
(113,0.10958904109589)
(114,0.109090909090909)
(115,0.108597285067873)
(116,0.108108108108108)
(117,0.112107623318386)
(118,0.111607142857143)
(119,0.111111111111111)
(120,0.110619469026549)
(121,0.114537444933921)
(122,0.114035087719298)
(123,0.11353711790393)
(124,0.11304347826087)
(125,0.112554112554113)
(126,0.112068965517241)
(127,0.11587982832618)
(128,0.115384615384615)
(129,0.114893617021277)
(130,0.114406779661017)
(131,0.113924050632911)
(132,0.117647058823529)
(133,0.117154811715481)
(134,0.116666666666667)
(135,0.116182572614108)
(136,0.119834710743802)
(137,0.119341563786008)
(138,0.118852459016393)
(139,0.118367346938776)
(140,0.121951219512195)
(141,0.121457489878543)
(142,0.120967741935484)
(143,0.120481927710843)
(144,0.12)
(145,0.119521912350598)
(146,0.123015873015873)
(147,0.122529644268775)
(148,0.122047244094488)
(149,0.12156862745098)
(150,0.12109375)
};

\addplot[red,very thick] coordinates {
(50,8.97435897435897E-02)
(51,0.089171974522293)
(52,8.86075949367089E-02)
(53,9.43396226415094E-02)
(54,0.09375)
(55,0.093167701863354)
(56,9.25925925925926E-02)
(57,9.20245398773006E-02)
(58,9.75609756097561E-02)
(59,0.096969696969697)
(60,9.63855421686747E-02)
(61,0.101796407185629)
(62,0.101190476190476)
(63,0.100591715976331)
(64,0.105882352941176)
(65,0.105263157894737)
(66,0.104651162790698)
(67,0.109826589595376)
(68,0.109195402298851)
(69,0.108571428571429)
(70,0.107954545454545)
(71,0.107344632768362)
(72,0.112359550561798)
(73,0.111731843575419)
(74,0.111111111111111)
(75,0.110497237569061)
(76,0.115384615384615)
(77,0.114754098360656)
(78,0.114130434782609)
(79,0.118918918918919)
(80,0.118279569892473)
(81,0.117647058823529)
(82,0.122340425531915)
(83,0.121693121693122)
(84,0.121052631578947)
(85,0.120418848167539)
(86,0.125)
(87,0.124352331606218)
(88,0.123711340206186)
(89,0.123076923076923)
(90,0.127551020408163)
(91,0.126903553299492)
(92,0.126262626262626)
(93,0.125628140703518)
(94,0.13)
(95,0.129353233830846)
(96,0.128712871287129)
(97,0.133004926108374)
(98,0.132352941176471)
(99,0.131707317073171)
(100,0.131067961165049)
(101,0.135265700483092)
(102,0.134615384615385)
(103,0.133971291866029)
(104,0.133333333333333)
(105,0.137440758293839)
(106,0.136792452830189)
(107,0.136150234741784)
(108,0.14018691588785)
(109,0.13953488372093)
(110,0.138888888888889)
(111,0.142857142857143)
(112,0.142201834862385)
(113,0.141552511415525)
(114,0.140909090909091)
(115,0.14027149321267)
(116,0.144144144144144)
(117,0.143497757847534)
(118,0.142857142857143)
(119,0.146666666666667)
(120,0.146017699115044)
(121,0.145374449339207)
(122,0.144736842105263)
(123,0.148471615720524)
(124,0.147826086956522)
(125,0.147186147186147)
(126,0.150862068965517)
(127,0.150214592274678)
(128,0.14957264957265)
(129,0.148936170212766)
(130,0.152542372881356)
(131,0.151898734177215)
(132,0.151260504201681)
(133,0.150627615062762)
(134,0.154166666666667)
(135,0.153526970954357)
(136,0.152892561983471)
(137,0.156378600823045)
(138,0.155737704918033)
(139,0.155102040816327)
(140,0.154471544715447)
(141,0.157894736842105)
(142,0.157258064516129)
(143,0.156626506024096)
(144,0.156)
(145,0.159362549800797)
(146,0.158730158730159)
(147,0.158102766798419)
(148,0.161417322834646)
(149,0.16078431372549)
(150,0.16015625)
};

\addplot[blue,dashed,very thick] coordinates {
(50,7.69230769230769E-02)
(51,8.28025477707006E-02)
(52,8.22784810126582E-02)
(53,8.17610062893082E-02)
(54,0.08125)
(55,8.69565217391304E-02)
(56,8.64197530864197E-02)
(57,8.58895705521472E-02)
(58,8.53658536585366E-02)
(59,9.09090909090909E-02)
(60,9.03614457831325E-02)
(61,8.98203592814371E-02)
(62,8.92857142857143E-02)
(63,8.87573964497041E-02)
(64,9.41176470588235E-02)
(65,9.35672514619883E-02)
(66,9.30232558139535E-02)
(67,9.82658959537572E-02)
(68,9.77011494252874E-02)
(69,9.71428571428571E-02)
(70,9.65909090909091E-02)
(71,0.101694915254237)
(72,0.101123595505618)
(73,0.100558659217877)
(74,0.1)
(75,0.104972375690608)
(76,0.104395604395604)
(77,0.103825136612022)
(78,0.103260869565217)
(79,0.102702702702703)
(80,0.10752688172043)
(81,0.106951871657754)
(82,0.106382978723404)
(83,0.105820105820106)
(84,0.110526315789474)
(85,0.109947643979058)
(86,0.109375)
(87,0.113989637305699)
(88,0.11340206185567)
(89,0.112820512820513)
(90,0.112244897959184)
(91,0.116751269035533)
(92,0.116161616161616)
(93,0.115577889447236)
(94,0.115)
(95,0.114427860696517)
(96,0.118811881188119)
(97,0.118226600985222)
(98,0.117647058823529)
(99,0.117073170731707)
(100,0.121359223300971)
(101,0.120772946859903)
(102,0.120192307692308)
(103,0.119617224880383)
(104,0.123809523809524)
(105,0.123222748815166)
(106,0.122641509433962)
(107,0.126760563380282)
(108,0.126168224299065)
(109,0.125581395348837)
(110,0.125)
(111,0.124423963133641)
(112,0.128440366972477)
(113,0.127853881278539)
(114,0.127272727272727)
(115,0.131221719457014)
(116,0.130630630630631)
(117,0.130044843049327)
(118,0.129464285714286)
(119,0.128888888888889)
(120,0.132743362831858)
(121,0.13215859030837)
(122,0.131578947368421)
(123,0.131004366812227)
(124,0.134782608695652)
(125,0.134199134199134)
(126,0.133620689655172)
(127,0.133047210300429)
(128,0.136752136752137)
(129,0.136170212765957)
(130,0.135593220338983)
(131,0.135021097046413)
(132,0.138655462184874)
(133,0.138075313807531)
(134,0.1375)
(135,0.141078838174274)
(136,0.140495867768595)
(137,0.139917695473251)
(138,0.139344262295082)
(139,0.138775510204082)
(140,0.142276422764228)
(141,0.1417004048583)
(142,0.141129032258065)
(143,0.140562248995984)
(144,0.14)
(145,0.143426294820717)
(146,0.142857142857143)
(147,0.142292490118577)
(148,0.145669291338583)
(149,0.145098039215686)
(150,0.14453125)
};
\draw[green,ultra thick] ({axis cs:93,0}|-{rel axis cs:0,0}) -- ({axis cs:93,0}|-{rel axis cs:0,1});
\end{axis}
\end{tikzpicture}
\end{subfigure}
\vspace{0.5cm}

\begin{subfigure}{\textwidth}
	\centering
	\caption{LMP}
	\label{Fig3d}
	
\begin{tikzpicture}
\begin{axis}[width=0.98\textwidth, 
height=0.65\textwidth,
xmin=47.5,
xmax=152.5,
xlabel = Number of list mandates,
ylabel = Mandate share,
ymajorgrids,
legend entries={DVT$\quad$,PVT$\quad$,NVT},
legend style={at={(0.5,-0.15)},anchor = north,legend columns = 3},
scaled ticks=false,
yticklabel=\pgfmathparse{100*\tick}\pgfmathprintnumber{\pgfmathresult}\,\%,
yticklabel style={/pgf/number format/.cd,fixed,precision=2}
]

\addplot[black,dotted,very thick] coordinates {
(50,1.28205128205128E-02)
(51,1.27388535031847E-02)
(52,1.26582278481013E-02)
(53,1.88679245283019E-02)
(54,0.01875)
(55,1.86335403726708E-02)
(56,1.85185185185185E-02)
(57,1.84049079754601E-02)
(58,1.82926829268293E-02)
(59,1.81818181818182E-02)
(60,1.80722891566265E-02)
(61,1.79640718562874E-02)
(62,1.78571428571429E-02)
(63,1.77514792899408E-02)
(64,1.76470588235294E-02)
(65,1.75438596491228E-02)
(66,1.74418604651163E-02)
(67,1.73410404624277E-02)
(68,1.72413793103448E-02)
(69,1.71428571428571E-02)
(70,1.70454545454545E-02)
(71,2.25988700564972E-02)
(72,2.24719101123596E-02)
(73,2.23463687150838E-02)
(74,2.22222222222222E-02)
(75,2.20994475138122E-02)
(76,0.021978021978022)
(77,2.18579234972678E-02)
(78,2.17391304347826E-02)
(79,2.16216216216216E-02)
(80,0.021505376344086)
(81,2.13903743315508E-02)
(82,2.12765957446809E-02)
(83,2.11640211640212E-02)
(84,2.10526315789474E-02)
(85,2.09424083769634E-02)
(86,2.08333333333333E-02)
(87,2.07253886010363E-02)
(88,2.57731958762887E-02)
(89,2.56410256410256E-02)
(90,2.55102040816327E-02)
(91,2.53807106598985E-02)
(92,2.52525252525253E-02)
(93,2.51256281407035E-02)
(94,0.025)
(95,2.48756218905473E-02)
(96,2.47524752475248E-02)
(97,2.46305418719212E-02)
(98,2.45098039215686E-02)
(99,0.024390243902439)
(100,2.42718446601942E-02)
(101,2.41545893719807E-02)
(102,2.40384615384615E-02)
(103,2.39234449760766E-02)
(104,2.38095238095238E-02)
(105,0.023696682464455)
(106,2.83018867924528E-02)
(107,0.028169014084507)
(108,2.80373831775701E-02)
(109,0.027906976744186)
(110,2.77777777777778E-02)
(111,2.76497695852535E-02)
(112,2.75229357798165E-02)
(113,2.73972602739726E-02)
(114,2.72727272727273E-02)
(115,2.71493212669683E-02)
(116,0.027027027027027)
(117,2.69058295964126E-02)
(118,2.67857142857143E-02)
(119,2.66666666666667E-02)
(120,2.65486725663717E-02)
(121,0.026431718061674)
(122,2.63157894736842E-02)
(123,2.62008733624454E-02)
(124,3.04347826086957E-02)
(125,3.03030303030303E-02)
(126,3.01724137931035E-02)
(127,3.00429184549356E-02)
(128,2.99145299145299E-02)
(129,2.97872340425532E-02)
(130,2.96610169491525E-02)
(131,0.029535864978903)
(132,2.94117647058824E-02)
(133,2.92887029288703E-02)
(134,2.91666666666667E-02)
(135,0.029045643153527)
(136,2.89256198347107E-02)
(137,2.88065843621399E-02)
(138,2.86885245901639E-02)
(139,2.85714285714286E-02)
(140,2.84552845528455E-02)
(141,2.83400809716599E-02)
(142,2.82258064516129E-02)
(143,3.21285140562249E-02)
(144,0.032)
(145,3.18725099601594E-02)
(146,3.17460317460317E-02)
(147,3.16205533596838E-02)
(148,0.031496062992126)
(149,3.13725490196078E-02)
(150,0.03125)
};

\addplot[red,very thick] coordinates {
(50,1.92307692307692E-02)
(51,1.91082802547771E-02)
(52,1.89873417721519E-02)
(53,1.88679245283019E-02)
(54,0.01875)
(55,1.86335403726708E-02)
(56,2.46913580246914E-02)
(57,2.45398773006135E-02)
(58,0.024390243902439)
(59,2.42424242424242E-02)
(60,2.40963855421687E-02)
(61,2.39520958083832E-02)
(62,2.38095238095238E-02)
(63,2.36686390532544E-02)
(64,2.35294117647059E-02)
(65,2.33918128654971E-02)
(66,2.32558139534884E-02)
(67,0.023121387283237)
(68,2.29885057471264E-02)
(69,2.28571428571429E-02)
(70,2.84090909090909E-02)
(71,2.82485875706215E-02)
(72,2.80898876404494E-02)
(73,2.79329608938547E-02)
(74,2.77777777777778E-02)
(75,2.76243093922652E-02)
(76,2.74725274725275E-02)
(77,2.73224043715847E-02)
(78,2.71739130434783E-02)
(79,0.027027027027027)
(80,2.68817204301075E-02)
(81,2.67379679144385E-02)
(82,2.65957446808511E-02)
(83,2.64550264550265E-02)
(84,2.63157894736842E-02)
(85,0.031413612565445)
(86,0.03125)
(87,3.10880829015544E-02)
(88,3.09278350515464E-02)
(89,3.07692307692308E-02)
(90,3.06122448979592E-02)
(91,3.04568527918782E-02)
(92,3.03030303030303E-02)
(93,3.01507537688442E-02)
(94,0.03)
(95,2.98507462686567E-02)
(96,2.97029702970297E-02)
(97,2.95566502463054E-02)
(98,2.94117647058824E-02)
(99,3.41463414634146E-02)
(100,3.39805825242718E-02)
(101,3.38164251207729E-02)
(102,3.36538461538462E-02)
(103,3.34928229665072E-02)
(104,3.33333333333333E-02)
(105,0.033175355450237)
(106,3.30188679245283E-02)
(107,3.28638497652582E-02)
(108,3.27102803738318E-02)
(109,3.25581395348837E-02)
(110,3.24074074074074E-02)
(111,0.032258064516129)
(112,3.21100917431193E-02)
(113,3.19634703196347E-02)
(114,3.63636363636364E-02)
(115,3.61990950226244E-02)
(116,0.036036036036036)
(117,3.58744394618834E-02)
(118,3.57142857142857E-02)
(119,3.55555555555556E-02)
(120,3.53982300884956E-02)
(121,3.52422907488987E-02)
(122,3.50877192982456E-02)
(123,3.49344978165939E-02)
(124,3.47826086956522E-02)
(125,3.46320346320346E-02)
(126,3.44827586206897E-02)
(127,3.43347639484979E-02)
(128,3.84615384615385E-02)
(129,3.82978723404255E-02)
(130,0.038135593220339)
(131,3.79746835443038E-02)
(132,3.78151260504202E-02)
(133,3.76569037656904E-02)
(134,0.0375)
(135,0.037344398340249)
(136,3.71900826446281E-02)
(137,0.037037037037037)
(138,3.68852459016393E-02)
(139,0.036734693877551)
(140,3.65853658536585E-02)
(141,3.64372469635628E-02)
(142,4.03225806451613E-02)
(143,4.01606425702811E-02)
(144,0.04)
(145,3.98406374501992E-02)
(146,3.96825396825397E-02)
(147,3.95256916996047E-02)
(148,3.93700787401575E-02)
(149,3.92156862745098E-02)
(150,0.0390625)
};

\addplot[blue,dashed,very thick] coordinates {
(50,1.92307692307692E-02)
(51,1.91082802547771E-02)
(52,1.89873417721519E-02)
(53,1.88679245283019E-02)
(54,0.01875)
(55,1.86335403726708E-02)
(56,1.85185185185185E-02)
(57,1.84049079754601E-02)
(58,1.82926829268293E-02)
(59,1.81818181818182E-02)
(60,1.80722891566265E-02)
(61,2.39520958083832E-02)
(62,2.38095238095238E-02)
(63,2.36686390532544E-02)
(64,2.35294117647059E-02)
(65,2.33918128654971E-02)
(66,2.32558139534884E-02)
(67,0.023121387283237)
(68,2.29885057471264E-02)
(69,2.28571428571429E-02)
(70,2.27272727272727E-02)
(71,2.25988700564972E-02)
(72,2.24719101123596E-02)
(73,2.23463687150838E-02)
(74,2.22222222222222E-02)
(75,2.20994475138122E-02)
(76,0.021978021978022)
(77,2.18579234972678E-02)
(78,2.71739130434783E-02)
(79,0.027027027027027)
(80,2.68817204301075E-02)
(81,2.67379679144385E-02)
(82,2.65957446808511E-02)
(83,2.64550264550265E-02)
(84,2.63157894736842E-02)
(85,2.61780104712042E-02)
(86,2.60416666666667E-02)
(87,2.59067357512953E-02)
(88,2.57731958762887E-02)
(89,2.56410256410256E-02)
(90,2.55102040816327E-02)
(91,2.53807106598985E-02)
(92,2.52525252525253E-02)
(93,2.51256281407035E-02)
(94,0.03)
(95,2.98507462686567E-02)
(96,2.97029702970297E-02)
(97,2.95566502463054E-02)
(98,2.94117647058824E-02)
(99,2.92682926829268E-02)
(100,0.029126213592233)
(101,2.89855072463768E-02)
(102,2.88461538461538E-02)
(103,2.87081339712919E-02)
(104,2.85714285714286E-02)
(105,0.028436018957346)
(106,2.83018867924528E-02)
(107,0.028169014084507)
(108,2.80373831775701E-02)
(109,0.027906976744186)
(110,3.24074074074074E-02)
(111,0.032258064516129)
(112,3.21100917431193E-02)
(113,3.19634703196347E-02)
(114,3.18181818181818E-02)
(115,3.16742081447964E-02)
(116,3.15315315315315E-02)
(117,0.031390134529148)
(118,0.03125)
(119,3.11111111111111E-02)
(120,3.09734513274336E-02)
(121,3.08370044052863E-02)
(122,3.07017543859649E-02)
(123,3.05676855895196E-02)
(124,3.04347826086957E-02)
(125,3.46320346320346E-02)
(126,3.44827586206897E-02)
(127,3.43347639484979E-02)
(128,3.41880341880342E-02)
(129,3.40425531914894E-02)
(130,3.38983050847458E-02)
(131,3.37552742616034E-02)
(132,3.36134453781513E-02)
(133,3.34728033472803E-02)
(134,3.33333333333333E-02)
(135,0.033195020746888)
(136,3.30578512396694E-02)
(137,3.29218106995885E-02)
(138,3.27868852459016E-02)
(139,3.26530612244898E-02)
(140,0.032520325203252)
(141,3.23886639676113E-02)
(142,3.62903225806452E-02)
(143,0.036144578313253)
(144,0.036)
(145,3.58565737051793E-02)
(146,3.57142857142857E-02)
(147,3.55731225296443E-02)
(148,3.54330708661417E-02)
(149,3.52941176470588E-02)
(150,0.03515625)
};
\draw[green,ultra thick] ({axis cs:93,0}|-{rel axis cs:0,0}) -- ({axis cs:93,0}|-{rel axis cs:0,1});
\end{axis}
\end{tikzpicture}
\end{subfigure}
\end{figure}

FIDESZ-KDNP's mandate share always exceed its vote share on party lists, while the other three parties are under-represented. Similarly, formula DVT is the best for FIDESZ-KDNP, PVT is the worst, and NVT means a kind of middle path. Other parties preferences is given by PVT $\succ$ NVT $\succ$ DVT. In the case of the smallest party (LMP) there are some ties between the methods due to its few mandates. Indivisibility of parliamentary seats is also responsible for non-monotinicity of the curves: only FIDESZ-KDNP loses from increasing the role of list votes, however, the allocation of further list seats causes some fluctuation in the plots. It has a larger effect for smaller parties.

On Figure \ref{Fig3a}, the level of two-thirds majority is also drawn (by a horizontal dashdotted line) since some laws require this threshold to be accepted or modified. One can see that the change of transfer formula from PVT to NVT was a crucial step for the incumbent FIDESZ-KDNP to secure a two-thirds (super)majority in the 2014 elections.\footnote{~The effects of the modified calculation of correction votes were highlighted by some Hungarian commentators, see, for example, \url{http://igyirnankmi.blog.hu/2014/04/09/ketharmad_csak_a_gyozteskompenzacioval} and \url{http://ideaintezet.blog.hu/2014/04/13/dontott_a_toredek}.}
According to the former rule PVT, it could be achieved only if the number of list seats do not exceed $75$.

An interesting point is that the $93$th list seat was allocated to FIDESZ-KDNP, with $92$ list seats the alliance has exactly a two-third share. On the other side, the $94$th list seat should be given to LMP, eliminating the two-third majority of FIDESZ-KDNP and increasing by $20$\% the number of mandates of LMP. It shows that, especially for small parties, a lot of randomness may influence whether a candidate gets a mandate or not.\footnote{~A possible solution to seat's indivisibility can be time-sharing, similarly to the stable roommates problem \citep{GaleShapley1962, Tan1991}. We think it may be worth to introduce partial mandates, namely, a candidate can get a mandate only for one or two years. Naturally, it is impossible to swap parliamentary members too often but the doubling or quadrupling of mandates still can significantly improve on this disproportion.}

Table \ref{Table4} shows the number of mandates for each party under different vote transfer formulas if there are $93$ list seats as specified by the election law. NVT gives approximately the average of DVT and PVT.

\begin{table}[!ht]
  \centering
  \caption{Seat allocation with $93$ list mandates, 2014 Hungarian parliamentary election}
  \label{Table4}
  \noindent\makebox[\textwidth]{
    \begin{tabularx}{\textwidth}{l CCC}
    \toprule
    Party  & DVT & PVT & NVT \\ \midrule
    FIDESZ-KDNP & 140   & 127   & 133 \\
    MSZP-EGYÜTT-DK-PM-MLP & 35    & 41    & 38 \\
    JOBBIK & 19    & 25    & 23 \\
    LMP   & 5     & 6     & 5 \\ \bottomrule
    \end{tabularx} }
\end{table}

Finally, we have examined whether a trade-off exists between the correction mechanism and the share (number) of list seats. Difference of seat allocations is measured by the formula $\sum_i(m_i - p_i)^2$ where $m_i$ is the mandate share of party $i$ according to the actual scenario and $p_i$ is the mandate share of party $i$ according to the true 2014 result.

These differences are plotted on Figure \ref{Fig4}. All methods have a unique global minimum, which is zero for NVT by definition. Every curve is similar to a parabola, however, a rise of list seats leads to a smaller discrepancy than a drop because of the fixed number of constituencies. Fluctuations are due to the indivisibility of parliamentary seats. In the case of $93$ list seats, the difference of DVT and PVT compared to the officially applied NVT is about the same.

\begin{figure}[htbp]
	\centering
	\caption{Difference of seat allocations compared to the official procedure}
	\label{Fig4}
	
\begin{tikzpicture}
\begin{axis}[width=0.98\textwidth, 
height=0.65\textwidth,
xmin=47.5,
xmax=152.5,
ymin=0,
xlabel = Number of list mandates,
ylabel = Difference of seat allocation,
legend entries={DVT$\quad$,PVT$\quad$,NVT},
legend style={at={(0.5,-0.15)},anchor = north,legend columns = 3},
scaled ticks=false,
yticklabel style={/pgf/number format/.cd,fixed,precision=4}
]

\addplot[black,dotted,very thick] coordinates {
(50,1.42547714736548E-02)
(51,1.28896948325317E-02)
(52,1.32954748694336E-02)
(53,1.23793360986193E-02)
(54,1.27765772171157E-02)
(55,1.15442402515067E-02)
(56,1.02899209586668E-02)
(57,1.06586650215591E-02)
(58,1.10296087620354E-02)
(59,9.9163721946434E-03)
(60,8.80206427234761E-03)
(61,9.14464390491937E-03)
(62,8.15770888879613E-03)
(63,8.48447331315121E-03)
(64,8.8141110693353E-03)
(65,7.8142575366325E-03)
(66,6.92495630665482E-03)
(67,7.2268929617608E-03)
(68,7.53205022183745E-03)
(69,6.64735553472417E-03)
(70,5.84676739839402E-03)
(71,5.34896220525265E-03)
(72,5.61654613085181E-03)
(73,5.88784867135251E-03)
(74,5.17069964251432E-03)
(75,4.45424921226363E-03)
(76,4.6989063694436E-03)
(77,4.0733378063762E-03)
(78,4.3043921839632E-03)
(79,3.67832016991845E-03)
(80,3.90089401696362E-03)
(81,3.34249665593742E-03)
(82,3.55209266073264E-03)
(83,3.00865555403914E-03)
(84,3.20999756504652E-03)
(85,2.7129834484554E-03)
(86,2.90195453774666E-03)
(87,3.09548862958021E-03)
(88,2.80366670289888E-03)
(89,2.36565792431709E-03)
(90,1.93064536998964E-03)
(91,2.09034538470456E-03)
(92,1.72040818932962E-03)
(93,1.86863968081614E-03)
(94,1.49948864927653E-03)
(95,1.64055327663849E-03)
(96,1.3176566631524E-03)
(97,1.44780716071065E-03)
(98,1.13918318114386E-03)
(99,1.26240570331531E-03)
(100,9.82686166696796E-04)
(101,1.09552611331262E-03)
(102,8.42589136921859E-04)
(103,9.48742801036222E-04)
(104,7.08700809852656E-04)
(105,8.04980558182038E-04)
(106,7.00157801373792E-04)
(107,5.14619988265556E-04)
(108,5.92552558217694E-04)
(109,4.09618666588619E-04)
(110,4.81624982520257E-04)
(111,3.29452698962839E-04)
(112,3.92354311651196E-04)
(113,2.54589108307369E-04)
(114,3.11783595079929E-04)
(115,1.90336600257663E-04)
(116,2.38876380074016E-04)
(117,1.42832179249185E-04)
(118,1.85881391046609E-04)
(119,9.2754360548945E-05)
(120,1.2757750376119E-04)
(121,7.01019332437702E-05)
(122,1.67073218947799E-05)
(123,3.26457727785231E-05)
(124,5.20933552259402E-05)
(125,6.61051362578316E-05)
(126,3.95283244657972E-05)
(127,3.83832040189684E-05)
(128,3.430652196968E-05)
(129,3.61084845324873E-05)
(130,3.15503510632608E-05)
(131,2.61604681977343E-05)
(132,4.62196384502109E-05)
(133,3.63714535572109E-05)
(134,5.21477782547582E-05)
(135,3.54642046415546E-05)
(136,8.49407544166729E-05)
(137,1.31889357112143E-04)
(138,9.8545813625556E-05)
(139,7.10872712912294E-05)
(140,1.4785682035866E-04)
(141,2.11840360956605E-04)
(142,1.68235609886033E-04)
(143,2.6471615533313E-04)
(144,2.20222317618241E-04)
(145,3.02920041727113E-04)
(146,4.19265474951899E-04)
(147,3.59313029343689E-04)
(148,3.05263836203793E-04)
(149,4.02580135156529E-04)
(150,3.42857102872639E-04)
};

\addplot[red,very thick] coordinates {
(50,4.43000662376435E-03)
(51,3.73605205986266E-03)
(52,3.99565096220312E-03)
(53,3.30306196262147E-03)
(54,2.72161490555794E-03)
(55,2.94563594956797E-03)
(56,2.58181787657443E-03)
(57,2.10749712289912E-03)
(58,1.60562985500774E-03)
(59,1.78129829437992E-03)
(60,1.40151091768054E-03)
(61,1.00937831237203E-03)
(62,1.15094889836779E-03)
(63,8.57669429231065E-04)
(64,5.66297607562774E-04)
(65,6.74492503058626E-04)
(66,4.60440651037637E-04)
(67,2.61664870017378E-04)
(68,3.37270669431476E-04)
(69,1.95832613998101E-04)
(70,1.41588611899566E-04)
(71,1.90545396512986E-04)
(72,8.38961367835035E-05)
(73,2.78944261713692E-05)
(74,4.63122735708449E-05)
(75,3.79182979012081E-05)
(76,1.9671528410763E-05)
(77,8.40146111859795E-06)
(78,5.64274073390586E-05)
(79,1.06920439344477E-04)
(80,6.68318109166225E-05)
(81,1.66675889489726E-04)
(82,2.80585737251491E-04)
(83,2.12558104781713E-04)
(84,3.60013573064448E-04)
(85,4.92787322406994E-04)
(86,6.81379003827063E-04)
(87,5.75168701933171E-04)
(88,7.77914603043667E-04)
(89,6.62203109149864E-04)
(90,8.87319620040388E-04)
(91,1.12931875654974E-03)
(92,9.88506665590118E-04)
(93,1.26259437892983E-03)
(94,1.54069467942729E-03)
(95,1.37562128020425E-03)
(96,1.68391300120971E-03)
(97,2.00630349347496E-03)
(98,1.81779732066114E-03)
(99,2.06212111559384E-03)
(100,2.40936902011934E-03)
(101,2.78335098764344E-03)
(102,2.56254245380212E-03)
(103,2.93803957450315E-03)
(104,2.71054677436912E-03)
(105,3.10686982527625E-03)
(106,3.5083425485343E-03)
(107,3.25994749866022E-03)
(108,3.69084601695167E-03)
(109,4.11619916236115E-03)
(110,3.84764920514229E-03)
(111,4.31058070515148E-03)
(112,4.75788197521944E-03)
(113,4.46990685129348E-03)
(114,4.79307434635697E-03)
(115,5.26577418827145E-03)
(116,5.76423060840257E-03)
(117,5.44918052699245E-03)
(118,5.94049950713727E-03)
(119,6.46535812008738E-03)
(120,6.13262033873644E-03)
(121,6.64102247942399E-03)
(122,6.3043010306664E-03)
(123,6.84058374834888E-03)
(124,7.36464856986245E-03)
(125,7.0113088444975E-03)
(126,7.57049469018261E-03)
(127,8.1089076844402E-03)
(128,8.48880782533066E-03)
(129,8.11197900963423E-03)
(130,8.69686052733965E-03)
(131,9.25029749687581E-03)
(132,8.85837589320751E-03)
(133,9.42699854429882E-03)
(134,1.00277200842291E-02)
(135,9.62128002813325E-03)
(136,1.02009191849204E-02)
(137,1.08192062783881E-02)
(138,1.03988041599122E-02)
(139,1.09884814056448E-02)
(140,1.05661323476415E-02)
(141,1.11892140757335E-02)
(142,1.15969975256041E-02)
(143,1.21964413609324E-02)
(144,1.17545338754072E-02)
(145,1.23955732731227E-02)
(146,1.30031283519471E-02)
(147,1.25488677174368E-02)
(148,1.32037820955519E-02)
(149,1.38186802435416E-02)
(150,1.33525466639593E-02)
};

\addplot[blue,dashed,very thick] coordinates {
(50,8.13360019979705E-03)
(51,7.08844630827622E-03)
(52,7.42531329405143E-03)
(53,6.49253887528891E-03)
(54,6.81159606133684E-03)
(55,5.89972199260675E-03)
(56,5.09016819392862E-03)
(57,5.37650504027591E-03)
(58,4.64805793831542E-03)
(59,3.9039408921253E-03)
(60,4.15770734472795E-03)
(61,3.72942193108952E-03)
(62,3.97615770886017E-03)
(63,3.36521785023388E-03)
(64,2.76344250158305E-03)
(65,2.97868471338353E-03)
(66,2.46342474402828E-03)
(67,1.97023108980286E-03)
(68,2.15438902644311E-03)
(69,1.72708069111839E-03)
(70,1.89655403820465E-03)
(71,1.48752793119796E-03)
(72,1.14113283155645E-03)
(73,1.28068627154411E-03)
(74,1.4269810495495E-03)
(75,1.09445092287093E-03)
(76,8.02173580999511E-04)
(77,9.19523406057841E-04)
(78,7.69171182927774E-04)
(79,5.50041816867592E-04)
(80,3.43390927461304E-04)
(81,4.21316359994036E-04)
(82,2.66083620511158E-04)
(83,3.31852260395031E-04)
(84,1.84899776936996E-04)
(85,2.43190884858931E-04)
(86,1.27666430543982E-04)
(87,5.01132925336247E-05)
(88,8.20129457509347E-05)
(89,2.12600638774799E-05)
(90,4.22990151416772E-05)
(91,1.42067487796922E-05)
(92,4.09270721544759E-06)
(93,0)
(94,3.61720663619603E-05)
(95,3.00652532438715E-05)
(96,6.4207769827824E-05)
(97,9.59533735872574E-05)
(98,6.59585419927098E-05)
(99,1.32903343210211E-04)
(100,2.28846272532789E-04)
(101,1.75619337725429E-04)
(102,2.83209397990518E-04)
(103,2.21220663897382E-04)
(104,3.52725924085955E-04)
(105,2.85548646484108E-04)
(106,4.13927790514945E-04)
(107,5.91488272466874E-04)
(108,5.02575201787537E-04)
(109,6.66174566694888E-04)
(110,8.23914275349286E-04)
(111,7.18323938295692E-04)
(112,9.37120056036788E-04)
(113,1.14182590109934E-03)
(114,1.01592413247767E-03)
(115,1.27394715523357E-03)
(116,1.14389657071148E-03)
(117,1.36271025005123E-03)
(118,1.22603913127838E-03)
(119,1.46991175244671E-03)
(120,1.7544761917347E-03)
(121,1.59883079168724E-03)
(122,1.86922939671951E-03)
(123,1.70746405863583E-03)
(124,2.01331724404622E-03)
(125,2.24487439719116E-03)
(126,2.07183611789191E-03)
(127,2.35818460932137E-03)
(128,2.70361772110951E-03)
(129,2.51313826903672E-03)
(130,2.82218407122974E-03)
(131,2.62653258640417E-03)
(132,2.98904619266043E-03)
(133,3.31918491865234E-03)
(134,3.10700539969529E-03)
(135,3.49658101785311E-03)
(136,3.28188991521966E-03)
(137,3.61813461922348E-03)
(138,3.39894325906291E-03)
(139,3.75363997650828E-03)
(140,4.15726203268509E-03)
(141,3.92263646971097E-03)
(142,4.20147412295084E-03)
(143,4.57750118984262E-03)
(144,4.33198111158809E-03)
(145,4.75383447644255E-03)
(146,5.14590952918629E-03)
(147,4.88600058450926E-03)
(148,5.32945478910831E-03)
(149,5.06796781944495E-03)
(150,5.46260255866724E-03)
};
\draw[green,ultra thick] ({axis cs:93,0}|-{rel axis cs:0,0}) -- ({axis cs:93,0}|-{rel axis cs:0,1});
\end{axis}
\end{tikzpicture}
\end{figure}
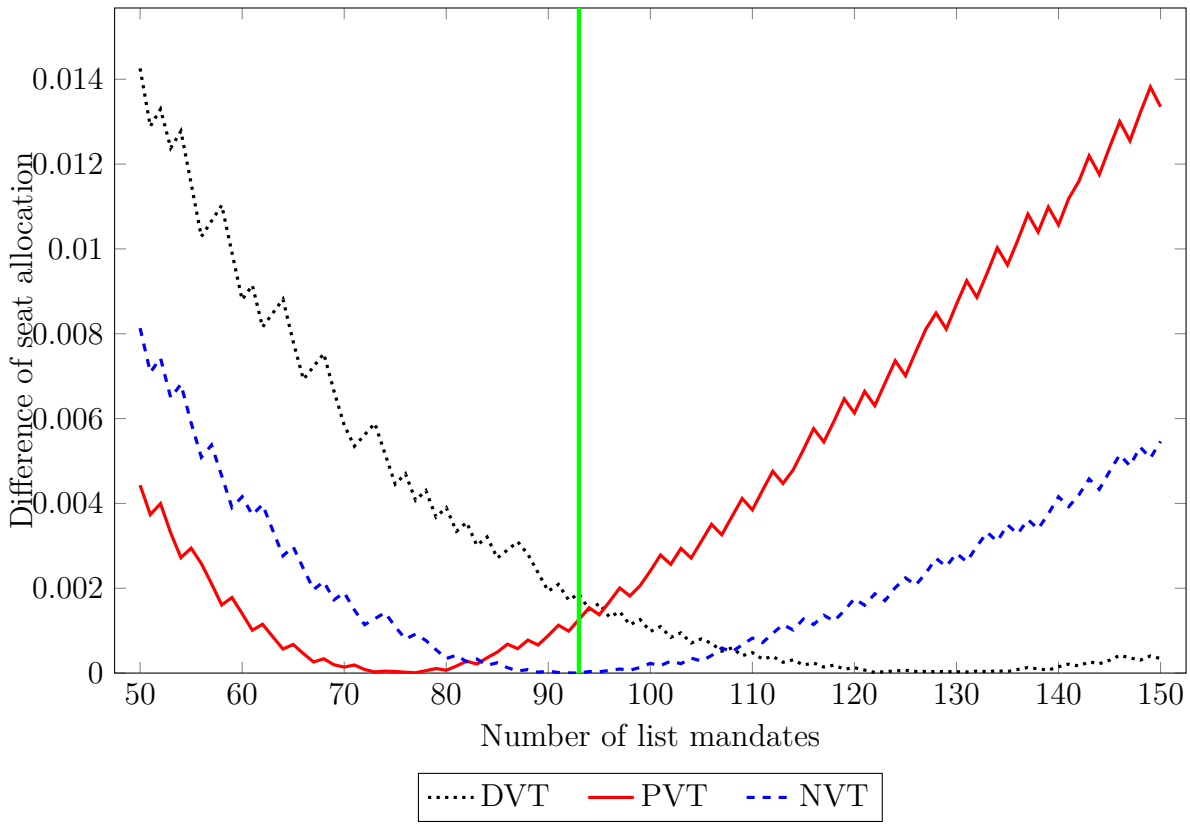

Scenarios resulting in minimal difference are detailed in Table \ref{Table5}. Thus party structure remains almost unchanged if formula PVT is applied with $77$ list seats or DVT is used with $122$ list seats.
Since an important motive behind the electoral reform was the reduction of the size of the parliament (visible in a cut of the number of constituencies), it is probable that the change of vote transfer formula PVT to NVT was other causes.

\begin{table}[htbp]
  \centering
  \caption{Seat allocation in different scenarios, 2014 Hungarian parliamentary election}
  \label{Table5}
  \noindent\makebox[\textwidth]{
    \begin{tabularx}{1\textwidth}{l CC >{\bfseries}C C} \toprule
    Party & DVT   & PVT   & NVT   & PVT \\ \midrule
    FIDESZ-KDNP & 153   & 122   & 133   & 138 \\
    MSZP-EGYÜTT-DK-PM-MLP & 43    & 35    & 38    & 52 \\
    JOBBIK & 26    & 21    & 23    & 35 \\
    LMP   & 6     & 5     & 5     & 8 \\ \midrule
    Number of list seats & 122   & 77    & 93    & 127 \\
    Number of mandates & 228   & 183   & 199   & 233 \\ \bottomrule
    \end{tabularx} }
\end{table}

The last column of Table \ref{Table5} presents another scenario which is close to the electoral system used between 1990 and 2010 with $233$ mandates ($106/233 \approx 176/386$) under the rule DVT.\footnote{~Note that it does not correspond to a hypothetical seat allocation under the former electoral law on the basis of 2014 results since there were other changes in the system.}
It leads to a parliament much moderately dominated by FIDESZ-KDNP since a substantial reduction of the ratio of constituency seats as well as the use of the vote transfer formula most favorable for small parties. It is worth to compare it to the variant with almost the same number of mandates ($228$) under DVT. 

\section{Summary} \label{Sec6}

A framework of an electoral system has been presented such that any outcome between majoritarian and proportional rules is achievable. It is based on vote transfer, incorporating single-member constituencies and a compensatory mechanism with three different methods, DVT, PVT and NVT for calculating list votes.
Investigation of the model have led to some interesting theoretical and practical results.

In a two-party system there exists a trade-off between certain correction rules. The party with a majority of votes has an incentive to increase the share of constituency seats in order to secure a bigger majority. However, it is a risky strategy: if the distribution of votes among the districts is favorable for the other party, it may lead to a loss of the election. It is true according to simulations as well as to exact mathematical formulas derived in extreme scenarios when a party gerrymanders district boundaries.
In order to mitigate this danger the dominant party has a tool, that is, to choose the vote transfer mechanism NVT.

\begin{observation} \label{Obs1}
In a two-party system NVT may be preferred over DVT by the dominant party due to an increased chance of obtaining majority despite the smaller expected vote share.
\end{observation}

The three transfer rules vary considerably with respect to their potential for manipulation. NVT is more immune to strategic behavior than PVT, while it is practically impossible under DVT.

\begin{observation} \label{Obs2}
From a strict mathematical perspective, the simplest method DVT is preferred to use in order to avoid the possibility of strategic manipulation arising from mixing constituency and list votes.
\end{observation}

Nevertheless, keep in mind that DVT immediately excludes to provide for a fully proportional outcome, which is possible under PVT and NVT, despite they (especially NVT) require a large number of compensation seats for it.

A scrutiny of 2014 Hungarian parliamentary election results reveals that the rule applied for the calculation of correction votes cannot be evaluated in itself, only together with the share of list seats in the system.

\begin{observation} \label{Obs3}
DVT, PVT and NVT may be functionally equivalent with an appropriate choice of the share of constituency mandates. Then, if the number of single member districts is fixed, NVT requires the fewest parliamentary seats.
\end{observation}

It remains an open question whether Observation \ref{Obs3} can be generalized, for example, if there are parties with a strong regional but modest national presence, so they are able to win some constituencies despite a small share in national votes. It may be the case in some socially or ethnically divided countries such as Afghanistan, Canada, Great Britain or Ukraine.

While it is apparent that the incumbent party FIDESZ-KDNP has benefited in the 2014 parliamentary election from the change of transfer vote formula (NVT instead of PVT) by the Constitution of 2012, it is clear that pure theory could not decide which rule is the best. It depends on whether the election system should be proportional or should create the conditions for a stable government.


\end{document}